\documentclass[a4paper,11pt]{article}

\usepackage{amsmath,amsfonts,amssymb,caption,graphicx}
\usepackage{ushort}
\usepackage{makeidx}
\usepackage{float}
\usepackage{accents}
\usepackage{color}
\usepackage{framed}
\usepackage{versions}
\usepackage{xcolor}
\usepackage{emptypage}
\usepackage{array}
\usepackage{multirow}
\usepackage{subcaption}

\usepackage{amsthm}

\usepackage{wrapfig}
\usepackage{booktabs}       
\def\hmath$#1${\texorpdfstring{{\rmfamily\textit{#1}}}{#1}}

\usepackage{tikz}
\usetikzlibrary{trees}
\newtheorem{theorem}{Theorem}
\newtheorem{lemma}{Lemma}
\newtheorem{corollary}{Corollary}
\usepackage[bookmarks]{hyperref}
\hypersetup{
    colorlinks=true,   	
    linkcolor=red,      
    citecolor = [rgb]{0 0.66 0},  
    filecolor=magenta, 	
    urlcolor=blue
}

\newcommand{\RL}{{\mathbb R}}

\newcommand{\IND}{{\mathbb I}}

\def\ba{\begin{align}}
\def\ea{\end{align}}
\def\ban{\begin{align*}}
\def\ean{\end{align*}}

\def\be{\begin{eqnarray}}
\def\ee{\end{eqnarray}}
\def\ben{\begin{eqnarray*}}
\def\een{\end{eqnarray*}}

\def\bqq{\begin{equation}}
\def\eqq{\end{equation}}
\def\bqqn{\begin{equation*}}
\def\eqqn{\end{equation*}}

\def\elabel#1{\label{e:#1}}

\def\sq{$\Box$}

\def\qed{\ifmmode\sq\else{\unskip\nobreak\hfil
\penalty50\hskip1em\null\nobreak\hfil\sq
\parfillskip=0pt\finalhyphendemerits=0\endgraf}\fi\par\medbreak}

\newsavebox{\junk}
\savebox{\junk}[1.6mm]{\hbox{$|\!|\!|$}}

\def\til={{\widetilde =}}

\def\clM{{\cal M}}

\def\clS{{\cal S}}
\def\clT{{\cal T}}

 \def\eq#1/{(\ref{#1})}

\def\eq#1/{(\ref{e:#1})}

\newcommand{\beqn}[1]{\notes{#1}%
\begin{eqnarray} \elabel{#1}}

\newcommand{\eeqn}{\end{eqnarray} } 

\newcommand{\beq}[1]{\notes{#1}%
\begin{equation}\elabel{#1}}

\newcommand{\eeq}{\end{equation}} 

\def\bdes{\begin{description}}
\def\edes{\end{description}}

\def\notes#1{}

\definecolor{mag}{rgb}{0.7,0,0.3}
\definecolor{dgreen}{rgb}{0.1,0.5,0.1}
\definecolor{dred}{rgb}{.8,0,0}
\definecolor{gray}{rgb}{.8,.8,.8}
\definecolor{brown}{rgb}{0.6451,0.3706,0.1745}

\begin{document}

\title{ \vspace*{-0.2in}
The Bayesian Context Trees State Space Model \\
for time series modelling and forecasting}

 \author
 {
 	 Ioannis Papageorgiou
     \thanks{ Department of Engineering,
         University of Cambridge,
        Trumpington Street, Cambridge CB2 1PZ, UK.
                 Email: \texttt{\href{mailto:ip307@cam.ac.uk}%
			 {ip307@cam.ac.uk}}.
 	 }
  \and
     Ioannis Kontoyiannis
     \thanks{Statistical Laboratory, Centre for Mathematical Sciences, University of Cambridge, Wilberforce Road, Cambridge CB3 0WB, UK.
                 Email: \texttt{\href{mailto: yiannis@maths.cam.ac.uk}%
			 { yiannis@maths.cam.ac.uk}}.
         }
 }

\date{}
\maketitle



\vspace*{-0.25in}

\vspace*{-0.3 cm}

\begin{center}
    \large \textbf{Abstract}
\end{center}

A hierarchical Bayesian framework is introduced for developing tree-based mixture models for time series, partly motivated by applications in finance and forecasting. At the top level, meaningful discrete states are identified as appropriately quantised values of some of the most recent samples. At the bottom level, a different, arbitrary `base’ model is associated with each state. This defines a very general framework that can be used in conjunction with any existing model class to build flexible and interpretable mixture models. We call this the Bayesian Context Trees State Space Model, or the BCT-X framework. \textcolor{black}{Appropriate algorithmic tools are described, which allow for  effective and efficient Bayesian inference and learning; these algorithms can be updated sequentially, facilitating  online forecasting}. The utility of the general framework is illustrated in the particular instances when AR or ARCH models are used as base models. The latter results in a mixture model that offers a powerful way of modelling the well-known volatility asymmetries in financial data, revealing a novel, important feature of stock market index data, in the form of an enhanced leverage effect. In forecasting, the BCT-X methods are found to outperform several state-of-the-art techniques, both in terms of accuracy and computational~requirements.

\medskip


\vspace*{-0.2 cm}

\section{Introduction}
Time series modelling, inference and forecasting 
are well-studied tasks in statistics and machine learning (ML), 
with critical applications throughout finance,
the sciences, and engineering. 
A wide range of approaches exist, from classical statistical 
methods~\cite{box:book,durbin:book2},
to modern ML techniques: notably  
neural networks~\cite{benidis:20,alexandrov:20,zhang:98},
Gaussian processes (GPs)~\cite{rasmussen:book,roberts:13,frigola:15}, 
and 
matrix factorisations~\cite{yu:16,faloutsos:18}.
{
Recent reviews that provide a thorough comparison between ML and classical methods specifically in the time series setting can be found in~\cite{makridakis:18b,makridakis:18,ahmed:10}.
Motivated in part by important applications
in financial time series and the inherent limitations of neural network models -- in that the they lack
interpretability and have large training-data requirements~\cite{makridakis:18b,benidis:20} -- 
in this work
we propose a new, general 
class of flexible hierarchical Bayesian models, which are both 
naturally {\em interpretable} and suitable for applications with 
limited training data. For these models, we provide computationally 
efficient -- linear complexity~-- algorithms for \textcolor{black}{effective} inference,
forecasting, and posterior exploration. 
}


Roughly speaking, time series models can be broadly categorised in two classes, 
depending on the absence or presence of an underlying hidden state process. 
The first class includes the family of linear autoregressive (AR) 
and autoregressive integrated moving average
(ARIMA) models~\cite{box:book,tsay:05,tong:book},
which directly model the stochastic mapping from previous 
to current observations, as well as
numerous AR extensions and generalisations,
including nonlinear AR models that employ
GPs~\cite{roberts:13,girard:02,turner:12,murray:01} 
and neural networks to model nonlinear dynamics. Feed-forward neural 
networks,
including the Neural Network AR (NNAR) model,
have been used since the 1990s in this setting~\cite{zhang:98},
and recurrent neural networks (RNN) have also been 
employed~\cite{graves:13,salinas:20,oreshkin:19}, with the DeepAR 
model of~\cite{salinas:20} that uses Long Short-Term Memory (LSTM) 
cells~\cite{hochreiter:97} being one of the most successful approaches.

The second class consists of State Space Models 
(SSM)~\cite{durbin:book2,hyndman:book}, which can describe more 
complex dynamics by combining an underlying hidden state process 
(`state-transition') with an observation model (`emission'). 
Classical approaches here include hidden Markov 
models (HMMs)~\cite{cappe:book06},
the linear Gaussian SSM~\cite{kalman:60},
exponential smoothing 
methods~\cite{hyndman:book},
\textcolor{black}{and dynamic factor models~\cite{stock2011dynamic,forni2000generalized,stock2002macroeconomic}}.
Again, various extensions have been proposed by making the transition 
or emission equations nonlinear, using  
GPs~\cite{frigola:13,frigola:14,eleftheriadis:17,turner:12} or
neural 
networks~\cite{krishnan:15,krishnan:17,zheng:17,karl:17,%
rangapuram:18}. Inference in these 
settings is often a challenging task, requiring sophisticated and 
computationally intense approximate techniques, including 
Particle MCMC~\cite{doucet:book,andrieu:10} and variational 
methods like stochastic gradient variational 
Bayes~(SGVB)~\cite{kingma:13,rezende:14}. 

\medskip

\noindent
{\bf The BCT-X framework.}
In this work, we introduce a novel Bayesian modelling approach 
that combines important features of both the above classes.
First, meaningful discrete states are identified,
but these are {\em observable} rather than hidden:
Each {\em state} is a short discrete
sequence given by the appropriately 
quantised values of some of the most recent samples.
This collection of states is described by a discrete
{\em context-tree model}. 
Then a different time series model -- a {\em base model} --
is associated with each possible current state,
generating the next sample.

In technical terms, we define a Bayesian hierarchical model.
The top level consists of a set of relevant states,
naturally represented as a discrete 
context-tree model, 
which admits a natural interpretation 
and captures important aspects of the 
structure present in the data.
And the bottom level associates an arbitrary time series 
model to each state.
 \textcolor{black}{Equivalently, at the top level we can think of the context-tree model as defining a partition of the state space: Every discrete state actually corresponds to a different region of this partition, for which at the bottom level we associate a different base model (that is to  be used in that region).  }
We call the resulting model class 
the {\em Bayesian Context Trees State Space Model},
or BCT-X. The `BCT' part refers to the 
discrete context-tree models, and the `X' 
indicates that any existing time series family
of models can be employed as base models 
for the different states.

\textcolor{black}{Bayesian Context Trees (BCT) were originally
introduced~\cite{BCT-JRSSB:22} as a Bayesian framework for modelling variable-memory Markov chains. So far, the BCT framework has only been used in the restricted setting of {\em{discrete-valued}} time series -- which of course 
considerably limits its practical applicability.
The main conceptual novelty of this work is that, 
although context trees are naturally suited for discrete data 
(and have only been used in this setting), 
we show that they can also be utilised in a 
very effective way for {\em{real-valued}}
time series.
Specifically, by introducing an appropriate quantiser, context-tree models can be used to define discrete states (or equivalently, partitions of the state space) as explained above, and hence provide a general way of building flexible and interpretable mixture models, by associating a different,
arbitrary base-model to each state/state-space region.
}

Although at times we refer to BCT-X as a `model', 
it is in fact a general framework for building Bayesian 
mixture models for time series, which can be used in conjunction 
with any existing model class. The resulting model family is 
rich, 
flexible, and much more general than the class one starts with. 
For example, using any of the standard linear families (like 
AR or ARIMA)
leads to much more general mixture models that
can capture highly nonlinear trends and 
are also easily interpretable. Moreover,
this type of observable state 
process 
facilitates effective Bayesian inference.
This is achieved by exploiting the structure of context-tree models,
in a similar spirit to the  methods developed for 
discrete-valued time series in~\cite{BCT-JRSSB:22,ctw-isit:21}.

\medskip

\noindent
{\bf Algorithmic tools for inference.}
\textcolor{black}{
A family of important algorithms are 
introduced in conjunction with the BCT-X framework, 
by adapting the algorithms described in~\cite{BCT-JRSSB:22} for the discrete-valued setting 
to the much more general setting considered here. 
First, the `Generalised Context-Tree Weighting' algorithm (GCTW) computes
the \textit{evidence}~\cite{mackay:92} 
exactly, with all models and parameters integrated out. 
Second, the `Generalised Bayesian Context Trees' (GBCT) and $k$-GBCT 
algorithms identify
the~$k$ {\em a~posteriori} most likely (MAP) context-tree models,
along with their  posterior probabilities.
Lastly, a direct Monte Carlo procedure is developed
for further exploring the~posterior.
}

Importantly, using this Bayesian approach, the set of relevant states --
namely, the context tree model --
is identified directly from the data, and does not need to be specified
{\em a priori}. This 
avoids common overfitting problems that
lead to lack of interpretability 
and poor out-of-sample performance.
Moreover, the above algorithms
have only linear complexity and allow 
for sequential updates. These properties
are ideally 
suited for online forecasting
and provide an important practical advantage 
compared to standard ML-based time series approaches. 

The BCT-X modelling framework along with these algorithmic 
tools provide a powerful Bayesian framework for
modelling and for \textcolor{black}{effective} -- and computationally
efficient -- Bayesian inference.
\textcolor{black}{The 
application of the general framework is illustrated
by introducing two particular model classes that 
are found to be useful in practical applications, described next.}

\medskip

\noindent
{\bf BCT-AR models.}
First, we examine the case where AR models are used as 
base models for BCT-X, 
with a different AR model associated to each state. 
We refer to the resulting model class 
as the {\em Bayesian Context Trees Autoregressive} 
(BCT-AR)~model. This is shown to be a flexible, nonlinear 
AR mixture model that 
generalises popular AR mixtures,
including the threshold AR~(TAR) 
models~\cite{tong:11,tong:80,tong:book} 
and the mixture AR~(MAR) models of~\cite{wong:00}. 
The BCT-AR model is found to outperform state-of-the-art methods 
in experiments on both simulated data and real-world data
from standard applications 
of nonlinear time series in economics and~finance, both in terms 
of forecasting accuracy and computational requirements. 

\medskip

\noindent
{\bf BCT-ARCH models.}
Second, we employ autoregressive conditional heteroscedastic 
(ARCH) models as base models. This results in 
yet another flexible 
mixture model class, BCT-ARCH, that provides a systematic and powerful way 
of modelling the well-known asymmetric response in volatility due to 
positive and negative shocks, which is an important feature of financial 
time series~\cite{tsay:05}. 
\textcolor{black}{In fact, modelling volatility asymmetries with the BCT-ARCH model serves as an important motivating example for illustrating the purpose and the practical utility of the general BCT-X methodology. 
Specifically, in contrast with traditional modelling of these asymmetries, 
the BCT-ARCH model identifies that it is not only the sign of the most recent
value-change that is relevant in the volatility response, but 
that the exact pattern of recent ‘ups’ and ‘downs’ is also important; 
this is perfectly captured by a collection
of discrete states in the BCT-ARCH model.} 
As a result, the BCT-ARCH model is found to outperform previous approaches 
that were developed to describe this effect,
and is able to identify this interesting and newly observed
structure in the data, in the form of an {\em enhanced leverage effect}.

\smallskip

Finally, we mention that a number of earlier approaches,
e.g.~\cite{alvarez:10,alvisi:07,berndt:94,fu:07,hu:14,liu:11,%
ouyang:10,hero:20}, have
employed discrete patterns in the analysis of~real-valued 
time series.
These works illustrate the fact that useful and~meaningful information
can indeed be extracted from discrete contexts. However,
in most cases the methods developed have been 
either application-
or task-specific, and often need to resort to {\em ad hoc} considerations
for inference. In contrast, in this work discrete states are used 
in a natural manner, by defining
a~hierarchical Bayesian modelling structure upon which principled Bayesian inference is~performed.

\vspace*{-0.1 cm}

\newpage

\section{BCT-X: The Bayesian Context Trees State Space Model} \label{bct}

\begin{wrapfigure}{r}{0.33\linewidth}
\vspace*{-0.28 cm}
  \begin{center}
\vspace*{-0.6 cm}
    \includegraphics[width= 0.73 \linewidth]{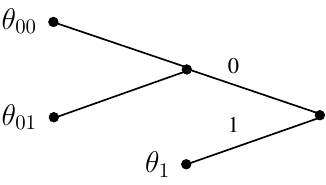}
  \end{center}
\vspace*{-0.5 cm}
\caption{\color{black} Example of a binary context-tree model $T$
 used for defining the set of states of BCT-X.}
\label{tree}
\vspace*{-0.2 cm}
\end{wrapfigure}
Before outlining the general construction of the
BCT-X framework for real-valued time series $(x_n)$, 
we briefly illustrate its structure through a particular 
example of the BCT-AR model with context depth $D=2$
and AR order $p$. 
Given the past values $(\ldots,x_{n-2},x_{n-1})$
and the context-tree model $T$ shown in Figure~\ref{tree},
the distribution of $x_n$ is determined 
as follows.

First, the {\em context} of $x_n$ is defined 
as the binary string given by $t:=(\IND_{\{x_{n-1}\geq 0\}},\IND_{\{x_{n-2}\geq 0\}})\in\{0,1\}^2$
of length $D=2$.
Then the corresponding {\em state} at time $n$ 
is the unique suffix $s$ of $t$ 
that appears as a leaf of the model~$T$.
For example, the context $t=(0,0)$ corresponds
to the state $s=00$, while the context
$t=(1,1)$ corresponds to the state $s=1$.
Finally, the distribution of $x_n$ 
given $(x_{n-p},\ldots,x_{n-1})$ and $s$ is determined
by the AR($p$) model with parameters~$\theta_s$.

Although at first glance this may seem quite
similar to the constructions
of the mixture autoregressive (MAR)~\cite{wong:00}
or threshold autoregressive (TAR)~\cite{tong:80} models, 
the BCT-X class is more general and critically different. First,
the set of relevant states is not defined 
by the modeller but it is determined by data:
A prior distribution is placed
on {\em all possible models $T$} of maximum
depth~$D$, and Bayesian inference is performed by considering the resulting posterior. Second, any model class can be used
for the {\em base models} in place of the AR models 
in this~example.

\subsection{Discrete contexts and states} 
\label{21}

Let $D\geq 0$ be a fixed, maximum context depth,
and let $x$ denote the sequence of observations, 
consisting of the time series $x_1^n=(x_1,x_2,\ldots,x_n)$
together with the initial context $x_{-D+1}^0=(x_{-D+1},\ldots,x_0)$,
with all $x_i\in\RL$. Throughout, we write 
$x_i^j=(x_i,x_{i+1},\ldots,x_j)$, for $i<j$.

The first step in the development of the BCT-X framework is the
construction of discrete, observable states. To that end,
we consider piecewise constant quantisers 
from $\mathbb{R}$ to a finite alphabet 
$A = \{0,1,\dots ,  m-1 \}$, of the form, 
\begin{equation}
\hspace*{-0.05 cm } Q(x) \!=\! \left\{
\begin{array}{ll}
\hspace*{-0.05 cm } 0, \ \;\; \;\; \;\;\; x < c_1, \\
\hspace*{-0.05 cm }  i , 
	\ \;\; \;\; \;\;\;\; c_{i} \leq x < c _{i+1} ,  \   1\leq i\leq m-2 ,\\
\hspace*{-0.05 cm } m-1, \   x \geq c_{m-1}, 
\end{array}
\right.  
\label{quant}
\end{equation}
where the alphabet size $m$, the thresholds 
$\{c_i\}$,
and the resulting quantiser $Q$ are 
considered to be fixed for the rest of this section. 
\textcolor{black}{
A core part of our general methodology is the development of
a systematic and computationally
efficient way to infer the thresholds $\{c_i\}$ from data,
as described in Section~\ref{hyp}.
The quantiser function $Q$ plays a key
role in the formulation of the BCT-X framework, 
as its role is crucial in describing the relevant states and state-space~partitions.}

 Given a maximum context depth $D\geq 0$, we consider
the class $\clT(D)$ of all proper $m$-ary trees of depth
no more than $D$, where
a tree $T$ is {\em proper} if any node in $T$ that is not a leaf has 
exactly $m$ children.
The elements of $T\in\clT(D)$
are the {\em context-tree models}:
Given $T\in\clT(D)$,
the sequence
$t=(Q(x_{i-1}),\ldots,Q(x_{i-D}))$
is the {\em context} at time $i\geq 1$.
Since $T$ is proper,
for any context $t$
there is a unique suffix
$s$ of $t$ that is a leaf of $T$.
This $s$ is the {\em state} at time $i$,
and it plays the role of a discrete feature vector 
that can be used to identify useful structure 
in the data. The leaves of $T$ define the set $\clS$ of discrete 
{\em states} in our hierarchical model. 
For example, for the tree of Figure~\ref{tree} and a binary quantiser with threshold $c=0$, 
we get $Q(x)=\IND_{\{x\geq 0 \}}$ and~$\mathcal {S} = \{ 1, 01, 00\}$.

\subsection{Base model class}

To complete the specification of the BCT-X framework, we associate 
a different time series model~$\mathcal {M}_s$ to each state $s$,
i.e., to each leaf $s$ of 
the context-tree model~$T$, giving a different conditional distribution for
the next sample.
We refer to the class of all such models $\clM _s$ as the {\em base 
model class}.
At time $i$, given the current state $s$ determined by
the context $t=(Q(x_{i-1}),\ldots,Q(x_{i-D}))$, the distribution 
of $x_i$ is given by the model $\mathcal {M}_s$ assigned to $s$.
Although general non-parametric models could also be used, 
for the rest of this paper we consider parametric models 
with parameters $\theta_s$ for each state $s$.
Altogether, a specific instance of the BCT-X framework consists of an
$m$-ary quantiser $Q$, a context-tree model $T\in\clT(D)$ that defines the set 
of discrete states $\clS$, a collection $\theta=\{\theta_s\}$
of parameter vectors~$\theta_s$ 
at the states $s\in\clS$, and a corresponding collection
$\{\clM_s\}=\{\clM(\theta_s)\}$ of base models.

Let $x$ denote a time series $x_1^n$ along with its
initial context $x_{-D+1}^0$.
Identifying $T$ with the collection of its leaves ${\cal S}$,
the likelihood induced by the resulting BCT-X model is,
\begin{align}\label{lik}
\begin{split}
 p(x|\theta,T) := p(x_1^n | T, \theta, x_{-D+1}^0) = \prod _{i=1} ^ n p(x_i| T, \theta, x_{-D+1}^{i-1} ) 
= \prod _{s \in T}\  \prod _ {i \in B_s} p(x_i| T, \theta _s , x_{-D+1}^{i-1} )  , 
\end{split} 
\end{align} 
where  $B_s$ is the set of indices $i\in\{1,2,\ldots ,n\}$ such that 
the state at time $i$ is $s$.

\subsection{Bayesian modelling} \label{23}

For the top level of the BCT-X framework, we consider 
context-tree models $T$ in the class
$\mathcal T (D)$, and employ
the Bayesian Context Trees (BCT) prior of~\cite{BCT-JRSSB:22}
on $\clT(D)$,  
\begin{equation}\label{prior}
\pi(T)=\pi_D(T;\beta)=\alpha^{|T|-1}\beta^{|T|-L_D(T)},\qquad T\in\clT(D), 
\end{equation}
where $\beta\in(0,1)$ is a hyperparameter, $\alpha =(1-\beta)^{1/(m-1)}$,
$|T|$ is the number of leaves of $T$,
and $L_D(T)$ is the number of leaves of
$T$ at depth $D$.
This prior penalises larger models
by an exponential amount to avoid overfitting; 
see~\cite{BCT-JRSSB:22} for an extensive
discussion of the properties of $\pi(T)$
and guidelines regarding the choice of $\beta$.
Given a context-tree model $T \in \mathcal T (D)$, 
the parameter prior on $\theta=\{\theta_s\}$
is a product of 
independent priors for each~$\theta_s$, 
$\pi (\theta | T) = \prod _ {s \in T} \pi (\theta_s)$.
The exact form of $\pi(\theta_s)$ of course depends
on the choice of the  base models
$\clM_s$.

\subsection{Bayesian inference} 
\label{23b}

\vspace*{-0.1 cm}

A typical obstacle to performing 
Bayesian inference is the difficulty 
in computing the 
normalising constant $p(x)$
of the posterior density (sometimes referred to as the 
{\em prior predictive likelihood} or simply as the 
{\em evidence}), which
in this case is given by, 
\begin{equation}
p(x)=\sum_{T \in \mathcal T (D)}\pi(T) p(x|T) = \sum_{T \in \mathcal T (D)}
\pi(T) \int _ {\theta} p(x|T,\theta) \pi (\theta | T) \  d \theta  . 
\label{ctw}
\end{equation} 
\textcolor{black}{
The power of the proposed Bayesian structure stems, in part, from the fact 
that, although $ \mathcal T (D)$ is enormously rich, consisting of 
doubly-exponentially many models in $D$, it is actually possible to
perform \textcolor{black}{effective} Bayesian inference within $\clT(D)$ very efficiently.
By appropriately modifying and extending the corresponding algorithms 
for discrete time series from~\cite{BCT-JRSSB:22},
we introduce the Generalised Context-Tree Weighting (GCTW) algorithm 
and the Generalised Bayesian Context Tree (GBCT) algorithm, 
which can be used in the present general setting 
of real-valued time series. We show that
GCTW computes $p(x)$ 
in~(\ref{ctw}) (Theorem~\ref{ctwth}), and GBCT identifies the 
MAP context-tree model (Theorem~\ref{bctth}). 
Theorems~\ref{ctwth} and~\ref{bctth} are proved in
Section~\ref{appA}
of the Supplemental Material,
where we also 
introduce the $k$-GBCT algorithm 
that obtains the top-$k$ \textit{a posteriori} 
most likely context-tree models; the details are omitted here.
}

Let $x$ be a time series $x_1^n$ with 
initial context $x_{-D+1}^0$.
For a model $T$ with associated parameters $\theta=\{\theta_s\}$,
the \textit{generalised estimated probabilities} $P_e(s,x)$ defined
at every node $s$ of $T$ 
play a central role in the
algorithms' descriptions and in the proofs of the theorems.
These are,
\begin{equation}\label{pes}
P_e(s,x) =  \int \prod _{i \in B_s}  p(x_i|T , \theta _s , x_{-D+1}^{i-1}) \ \pi (\theta _ s) \ d\theta_s . 
\end{equation}
Write
$q_i=Q(x_i)$ for the quantised samples.
Given a maximum context depth $D\geq 0$ and the value
of the hyperparameter $\beta$, the GCTW and GBCT
algorithms operate as follows.

\medskip

\noindent 
\textbf{GCTW: The generalised context-tree weighting algorithm}

\vspace*{-0.2 cm}

\begin{enumerate}
\item Build $T_{\text{MAX}}$, 
the smallest proper tree containing
all contexts $q_{i-D} ^ {i-1}, \  i=1,2,\ldots,n$,
as leaves.

\vspace*{-0.1 cm}

{\color{black} \item  Compute $P_e(s,x)$ as in~(\ref{pes}) for 
all node $s$ of $T_{\text{MAX}}$, including internal nodes.
}

\vspace*{-0.1 cm}

\item Starting at the leaves and proceeding recursively towards the 
root, at each node $s$ of $T_{\rm MAX}$ 
compute, \vspace*{-0.2 cm}
\[
P_{w,s}\!=\!
\left\{
\begin{array}{ll}
P_e(s,x),  \; \; \; &\mbox{if $s$ is a leaf,}\\
\beta P_e(s,x)+(1-\beta)\prod_{j=0}^{m-1} P_{w,sj},  \; \; \; &\mbox{otherwise,}
\end{array}\!\!
\right.\!\! 
\]  where $sj$ is the concatenation
of context $s$ and symbol $j$.
\end{enumerate}


\noindent 
\textbf{GBCT: The generalised Bayesian context tree algorithm}

\vspace*{-0.2 cm}

\begin{enumerate}
\item Build the tree $T_{\text{MAX}}$  as in GCTW.

\vspace*{-0.1 cm}

{\color{black} \item  Compute $P_e(s,x)$ as in~(\ref{pes}) for 
all node $s$ of $T_{\text{MAX}}$, including internal nodes.
}

\vspace*{-0.1 cm}

\item Starting at the leaves and proceeding recursively
towards the 
root, at each node $s$ of $T_{\rm MAX}$ compute,
\[
 P_{m,s}\!=\!
\left\{
\begin{array}{ll}
P_e(s,x), \;\;\; &\mbox{if $s$ is a leaf at depth $D$,}\\
\beta, \;\;\;   &\mbox{if $s$ is a leaf at depth $<D$,}\\
 \max\big\{\beta P_e(s,x),
(1-\beta)\prod_{j=0}^{m-1} \hspace*{-0.07 cm } P_{m,sj}\big\},  \ 
&\mbox{otherwise.}
\end{array}
\right. 
\] 
\item Starting at the root and proceeding recursively with its descendants,
for each node $s$ in $T_{\rm MAX}$: If the maximum above is achieved 
by the first term, prune all its descendants from $T_{\text{MAX}}$.
Let $T_1^*$ denote the resulting tree, after all nodes have been examined for pruning.
\end{enumerate}

\begin{theorem}\label{ctwth}
The weighted probability $P_{w,\lambda}$ at the root 
node $\lambda$ of $T_{\rm MAX}$ is equal to $p(x)$ as in~{\em (\ref{ctw})}.
\end{theorem}

\begin{theorem}\label{bctth}
For all $\beta\geq 1/2$, the
tree $T^*_1$ produced by 
the GBCT algorithm is the MAP context-tree model, 
$\pi(T_1^*|x)=\max_{T\in\clT(D)}\pi(T|x)$.
\end{theorem}

{\color{black} By introducing an appropriate quantiser $Q$ and 
using the discretised versions 
$q_i = Q(x_i)$ of the samples $x_i$,
those steps of the generalised versions of the algorithms that are 
associated with building the tree $T_{\rm MAX}$ and performing 
recursive operations on it remain the same as before. The important difference 
in the new algorithms comes from the generalised estimated 
probabilities $P_e(s,x)$ in~(\ref{pes}), computed in
Step 2 for every node of $T_{\rm MAX}$.
The necessary values that need to be stored at
every node $s$ in order to compute $P_e(s,x)$ heavily 
depend on the choice of the base models; see for example Section~\ref{bctar} and Section~\ref{s:Pest} for the specific quantities that we need to store for the BCT-AR and BCT-ARCH models, respectively. 

In fact, the discrete BCT framework of~\cite{BCT-JRSSB:22} can be 
viewd as a simple special case of the general BCT-X framework,
when simple categorical distributions are chosen as the base models 
associated to the leaves of the trees; other specific examples are 
the BCT-AR (Section~\ref{bctar}) 
and BCT-ARCH (Section~\ref{s:Pest}) models
examined in this paper.
What remains true in all cases, and is perhaps somewhat remarkable, 
is that in the general BCT-X setting, 
the new versions of the algorithms can be used to perform 
Bayesian inference in a dynamic programming fashion, 
for arbitrary base models associated at the leaves of the trees.}

\medskip

\noindent
{\bf Sampling from the posterior.} 
The GCTW, GBCT and $k$-GBCT algorithms
compute the evidence $p(x)$ and identify 
the top-$k$ MAP context-tree models.
Furthermore,
it is in fact possible to obtain independent samples from the 
model posterior $\pi(T|x)$ as follows. For every node $s$,
let $P_{b,s} = \beta P_e(s,x,) / P_{w,s}$.
Start with the tree consisting
of only the root node $\lambda$,
and with probability $P_{b,\lambda}$ 
mark it as a leaf and stop,
or, with probability $(1-P_{b,\lambda})$
add all~$m$ of its children to~$T$. 
Then, examine every new node~$s$ and either
mark it as a leaf with probability~$P_{b,s}$,
or add all~$m$ of its children to $T$ with probability
$(1-P_{b,s})$. 
Examining
all non-leaf nodes 
at depths $<D$ recursively until no more eligible nodes remain, 
produces a random tree $T \in \mathcal {T} (D)$.
Theorem~\ref{branchth}, proved in 
the Supplemental Material,
states that $T$ is a sample from the posterior.

\begin{theorem}\label{branchth}
The probability that the above branching procedure 
produces any particular context-tree 
model $T \in \mathcal {T} (D)$ is given by $\pi(T|x)$.
\end{theorem}

{\color{black}

The methodological tools developed in this section are based on the generalised form of the estimated probabilities $P_e(s,x)$ of~(\ref{pes}),
whose exact form depends on the particular choice of
base models.
In the next section, the general principle is illustrated
in
an interesting case where $P_e(s,x)$ can be computed 
explicitly, and the resulting mixture model is a flexible nonlinear 
model of practical interest. Specifically, an AR model ${\cal M}_s$ is
associated to each leaf~$s$, and we hence
refer to the resulting model class as the 
BCT-AR model;
the BCT-AR model is only one particular instance 
of the general BCT-X framework.

However, even when the
the integrals in~(\ref{pes}) are not available in closed form,
 the fact that the estimated
probabilities are in the form of standard 
marginal likelihoods makes it possible to 
compute them approximately by using 
standard methods; see, e.g.,~\cite{kass:95,chib:95,chib:01,friel:08,wood:11}. 
Then, the GCTW and GBCT algorithms 
can be used in exactly the same manner as before, with these approximations $\widehat {P_e} (s,x)$ in place of their exact values,
and therefore still facilitate effective Bayesian inference. This is illustrated in detail
via the BCT-ARCH version of BCT-X, where the estimated probabilities are not available
in closed form;
see Section~\ref{s:Pest}.
}

\subsection{Remarks on prior work}

Context-tree models -- introduced
as ``tree sources'' 
in the information-theoretic 
literature 
by Rissanen~\cite{rissanen:83b,rissanen:83,rissanen:86} 
in the 1980s --
have been employed 
widely in connection with various
problems on {\em discrete} data, especially
for data 
compression~\cite{weinbergeretal:94,willems-shtarkov-tjalkens:95,%
willems:98,matsushima:94,matsushima:09}.
In the statistics literature they were first
popularised as ``variable-length Markov 
chains''~\cite{buhlmann:00,buhlmann:04},
and more recently in connection with
the Bayesian Context Trees 
(BCT) framework~\cite{BCT-JRSSB:22,ctw-isit:21},
which has been 
found to be very effective 
in a range of statistical tasks,
including model selection, prediction, 
change-point detection and entropy 
estimation~\cite{papag-K-pre:23,branch-isit:22,lungu-arxiv:22,%
lungu-pap-K:22,papag-K-ITW:23}. 
A~central conceptual
novelty of the present work is in showing that
{\em discrete} context-tree models can in fact also be 
effectively utilised
for modelling \textit{real-valued} time 
series, by representing meaningful context-based discrete states that 
are used to build flexible and interpretable mixture models of practical 
interest.

\medskip

{\color{black} In more technical terms, apart from 
this crucial observation, the major contributions of this work compared 
to existing BCT papers are summarised below:

\begin{enumerate}
    \item We introduce a new, vastly richer and more flexible 
class of models for real-valued time series. These models are 
built by combining the use of context trees to define 
adaptive partitions of the state space, with the assignment of arbitrary 
base models to each resulting state-space region (corresponding to each leaf of the tree).

    \item We show that, despite the enormous size of these model
classes, efficient algorithmic tools can 
be developed -- including a principled Bayesian methodology 
for selecting appropriate quantisers -- for 
effective Bayesian inference within this general setting.

    \item We describe how the BCT-X framework naturally leads to 
effective Bayesian forecasting algorithms for real-valued time series.

    \item We illustrate the general methodology by taking,
as specific examples, AR and ARCH models as the base models.
We explicitly describe the resulting BCT-AR and BCT-ARCH model classes
and their properties, 
and we illustrate their superior forecasting performance on
a variety of relevant real-world applications.
\end{enumerate}
}

%
%
%
%
%

In a different but related direction,
starting with the algorithm of~\cite{breiman:book2}, 
the classification and regression trees (CART)
procedures 
have been widely used for regression and classification tasks.
Typically, these methods either rely
on greedy algorithms to a grow a tree with some stopping 
criteria -- sometimes with additional penalties for pruning~-- or
they adopt 
a Bayesian CART approach and place a prior on trees~\cite{chipman:98,bart}; 
see~\cite{loh:14} and~\cite{linero:17} for recent reviews.
The main difficulties in the use of these procedures are that
greedily-constructed trees tend to overfit the data, while the 
corresponding Bayesian approaches require the use of MCMC 
sampling, which is both 
computationally expensive and not guaranteed to be effective 
for inference. These reasons perhaps partly 
explain why, in applications, approaches 
exploiting CARTs for time series 
data~\cite{audrino:01,meek:02,dellaportas:07,taddy2011dynamic,deshpande2020vcbart} have 
not been used as widely as much simpler and more restricted model classes
like, e.g., threshold 
models~\cite{tong:11,zakoian:94}.

In contrast, in this work we take advantage
of the special sequential nature of time-series data 
which, together with the desirable properties of the BCT-X framework, 
allow for exact Bayesian inference in selecting the tree model 
-- in a very efficient manner.  
The BCT-X family of algorithms can identify  not only the 
MAP tree model (GBCT algorithm), but also the top-$k$ a posteriori most 
likely trees, for moderate $k$ ($k$-GBCT algorithm). 
Moreover, it is possible to further explore the posterior
on model space by easily obtaining i.i.d.\ samples
as described in Section~\ref{23b}, without the need to resort to MCMC; which of course offers another important practical advantage compared to Bayesian CART methods.

\newpage

\section{BCT-AR: The Bayesian Context Trees Autoregressive Model}
\label{bctar}

Given a quantiser $Q:\RL\to A=\{0,1,\ldots,m-1\}$, a fixed AR order $p$,
and a maximum context
depth $D\geq 0$, the BCT-AR model is the version of
BCT-X with AR$(p)$ models as the {\em base model}
class. {
\color{black}
With each possible state, that is, with each potential
leaf $s$ in a context-tree model
$T\in\clT(D)$,
we associate an AR$(p)$ model,
\begin{equation}\label{ar}
x_n = \phi _ {s,0} +  \phi _ {s,1} x_{n-1} + \dots + \phi _ {s,p} x_{n-p} + e_n = {\boldsymbol \phi _ s} ^{\text{T}} \ \mathbf{ \widetilde{ x} } _{n-1} + e_n , \quad  e_n \sim \mathcal N (0, \sigma _s ^2)  ,
\end{equation}
where $\boldsymbol \phi _ s = ( \phi _ {s,0}, \dots , \phi _ {s,p} ) ^ {\text{T}}$ 
and $ \mathbf{ \widetilde{ x} } _{n-1} = (1, x_{n-1},\dots, x_{n-p})^{\text{T}}$.
}
The parameters $\theta_s = (\boldsymbol \phi _ s , \sigma _ s ^2)$
of the model consist of
the AR coefficients $\boldsymbol\phi_s$ and the noise variance $\sigma_s^2$.

We place an inverse-gamma prior on $\sigma_s^2$ and a Gaussian prior 
on $\boldsymbol\phi_s$,
so that the joint prior on the parameters is 
$\pi(\theta _s) =\pi (\boldsymbol \phi _s | \sigma _s ^2) \pi (\sigma _ s ^2)$, with,
\begin{align} \label{ar_pr}
&\pi (\sigma _s ^2 ) = \text {Inv-Gamma} (\tau , \lambda), \quad \quad \pi (\boldsymbol \phi _s | \sigma _s ^2 )  = \mathcal N (\mu _o , \sigma_s ^2 \Sigma _o)  , 
\end{align}
where $(\tau, \lambda,  \mu _o ,  \Sigma _o  )$ are the prior hyperparameters.
This prior specification allows the exact computation of the estimated 
probabilities of~(\ref{pes}), and also gives closed-form posterior 
densities for 
the AR coefficients and the noise variance. These are given in 
Lemmas~\ref{lem:PeAR} and~\ref{lem:piAR}, whose
their proofs are given in Section~\ref{appB} of the Supplemental Material.
Importantly, since here the estimated probabilities $P_e(s,x)$ are available
in closed form, the GCTW, GBCT and $k$-GBCT algorithms can be used
directly for \textit{exact} Bayesian inference. 

\begin{lemma}
\label{lem:PeAR}
For the BCT-AR model, the estimated probabilities $P_e(s,x)$ 
as in~{\em (\ref{pes})} are given by,
\begin{equation}
P_e (s,x) = C_s ^ {-1} \ \frac{ \Gamma \left ( \tau + |B_s| / 2\right )  \ \lambda ^ \tau } {\Gamma (\tau ) \ \left ( \lambda +  D_s / 2  \right ) ^ {\tau + |B_s| / 2 } } \ , 
\end{equation}
where $|B_s|$ is the cardinality of $B_s$ in~{\em (\ref{lik})}, i.e., 
the number of observations  with  context~$s$, and, 
\[
C_s = \sqrt{ {(2 \pi )^{|B_s|  }}  \text {det}( I + \Sigma _o S_3 }),   \ 
D_s = s_1 +    \mu _o ^ {\text {T}} \Sigma _ o ^{-1}  \mu _o  -( \mathbf s_2  +    \Sigma _ o ^{-1}  \mu _o )^ {\text {T}} (S_3 + \Sigma _ o ^{-1} ) ^ {-1}  ( \mathbf s_2  +    \Sigma _ o ^{-1} \mu _o ),  \]
 \textit{with the sums $s_1, \mathbf s_2, S_3 $ defined as:}
\begin{equation}
s_1 = \sum _ {i \in B_s } x_i ^2 \quad , \quad \mathbf s_2 = \sum _ {i \in B_s } x_i \ \mathbf{ \widetilde{ x} } _{i-1} \quad , \quad S_3 = \sum _ {i \in B_s } \mathbf{ \widetilde{ x} } _{i-1}  \mathbf{ \widetilde{ x} } _{i-1} ^{\text{T}} \ .
\end{equation}
\end{lemma}

\begin{lemma}
\label{lem:piAR}
Given a tree model $T$, at each leaf $s$, the posterior distributions 
of the AR coefficients and the noise variance are given respectively by,
\begin{align} \label{10}
&\pi (\sigma _ s ^2 | T,x ) =  \text {Inv-Gamma } \left (\tau + |B_s| / {2} , \lambda + {D_s} / {2} \right )  , \quad \quad   
\pi (\boldsymbol \phi _s| T,x ) = t _ \nu (\mathbf m _s, P _s)   ,
\end{align}
where 
$ \nu = 2 \tau + |B_s| $, 
$t _ \nu$ denotes a multivariate $t$-distribution with $\nu$ 
degrees of freedom, and,
\begin{equation}
\mathbf m_s = (S_3 + \Sigma _o ^ {-1}) ^ {-1} (\mathbf s_2 + \Sigma _o ^ {-1}  \mu _ o) \ ,  \ 
 P_s ^ {-1} = \frac {2 \tau +|B_s| } {2 \lambda + D_s} (S_3 + \Sigma _o ^ {-1})  .\end{equation}
\end{lemma}

\begin{corollary}
The MAP estimators of $\boldsymbol \phi _s $ and $\sigma _ s ^2$ are given, respectively, by, 
\begin{equation}\label{map}
\widehat {\boldsymbol \phi_s} ^{\text {MAP}} = \mathbf m_s \  , \quad  \widehat {{\sigma _ s ^2 }}^ {\text {MAP}} = ( {2 \lambda + D_s}) / ( {2 \tau + |B_s| + 2} )    .
\end{equation}
\end{corollary}

\subsection{Computational complexity and sequential updates} 
\label{compl}

Consider executing the GCTW algorithm for a time series $x_1^n$ of length $n$. For each observation~$x_i$, $1 \leq i \leq n$, exactly $D+1$ nodes of $T_{\text{MAX}}$ need to be updated (or created if they do not already exist in $T_{\text{MAX}}$),
corresponding to the discrete contexts 
of length $0, 1,\ldots , D$ preceding~$x_i$.
For each one of these nodes, only the quantities $\{|B_s|, s_1, \mathbf s_2 , S_3\}$ need to be updated, which can be done efficiently by just adding an extra term to each sum. After repeating this step for every $1 \leq i \leq n$, using Lemma~\ref{lem:PeAR} the estimated probabilities $P_e (s,x)$ can be computed for all nodes of $T_{\text{MAX}}$, whose size is bounded as a function of the number of observations~$n$. And since the recursive step of GCTW only performs operations on $T_{\text{MAX}}$, it follows that its overall complexity as a function 
of $n$ is only~$\mathcal{O}(n)$, i.e., {\em linear} in the length 
of the time series. 
The same argument applies to 
GBCT as well, showing that both algorithms are computationally very
efficient and scale well with large time series.

{\color{black}
Since, as described above, the size of $T_{\text{MAX}}$ is bounded by $nD+1$, 
it is easy to see (taking into account the number of operations required 
to compute $P_e (s,x)$ for each node and the recursive steps of the algorithms)
that the complexity of GCTW and GBCT is in fact
$\mathcal{O}\left (nD (m+p^3) \right )$. Importantly, 
this is linear in each of $n$, $m$, and $D$,
and it is actually only slightly higher than the 
$\mathcal{O}\left (np^2\right )$
complexity of fitting a single AR model using least squares.
This means that \textit{exact} Bayesian inference can be performed 
within this much richer model class at an only 
marginally higher 
computational cost.


The above argument also shows that
both GCTW and GBCT 
can be updated {\em sequentially},
as for every additional observation $x_{n+1}$ only $D+1$ nodes of $T_{\text{MAX}}$ need to be updated.
In particular, 
sequential prediction can be performed efficiently.
Empirical running times for all forecasting experiments are reported 
in Section~\ref{run_times} of the Supplemental Material,
showing that the BCT-X methods are much more 
efficient than essentially all the alternatives
examined. The difference is quite large,
especially compared to  ML models that require 
heavy training and cannot be efficiently updated sequentially, giving 
empirical running times that are typically larger by several 
orders of magnitude; a general review or related issues
is given in~\cite{makridakis:18b}.

Finally, it is possible to show that the memory requirements of both algorithms are also linear in the length $n$ of the time series.
Arguing as before, since we only need to store the tree 
$T_{\text{MAX}}$ and the quantities $\{|B_s|, s_1, \mathbf s_2 , S_3\}$ 
for each node of $T_{\text{MAX}}$,
it is easy to see that the memory requirements of both algorithms 
are $\mathcal{O}\left (nDp^2 \right )$, again linear in both $n$ and $D$. 
}

\subsection{Choosing the hyperparameters, quantiser and AR order} 
\label{hyp}

The posterior distributions of $\boldsymbol \phi _s$ and $\sigma_ s ^2$ 
are not particularly sensitive to the prior hyperparameters. In all the 
experimental results below, the simple choice $\mu_o =0$ and 
$\Sigma_o =I$ in the AR coefficients' prior is made. 
In view of~(\ref{10}), $\tau$~and~$\lambda$ 
should be chosen to be relatively small in order to minimise 
their effect on the posterior while keeping the mode of the 
inverse-gamma prior, $\lambda/(\tau + 1)$, reasonable; setting 
$\lambda = \tau =1$ is a sensible choice when no 
additional prior information is available. For the context-tree model
prior, 
the default value \mbox{$\beta= 1-2^{-m+1}$}~\cite{BCT-JRSSB:22} 
and a maximum context depth $D=10$ are adopted. These default values
are used in all the experiments, without involving any hyperparameter tuning from data.

Finally, a principled Bayesian approach is taken for
selecting the quantiser thresholds $\{c_i\}$ of~(\ref{quant}) 
and the AR order $p$. Viewing 
them as extra parameters on an additional layer above everything else,
we place uniform priors on $\{c_i\}$ and $p$, and perform 
Bayesian model selection~\cite{rasmussen:00,mackay:92}
to obtain their MAP values.
The resulting posterior $p(\{c_i\},p|x)$ 
is proportional to the \textit{evidence}~$p(x|\{c_i\},p)$, which can 
be computed exactly using the GCTW algorithm.
Specifically, a suitable range of possible 
$\{c_i\}$ and $p$ is specified, and the values 
with the higher evidence are selected. 
For the AR order we take $1\leq p \leq p_{\text{max}}$ for an 
appropriate $p_{\text{max}}$ ($p_{\text{max}}=5$ in our experiments), 
and for the $\{c_i\}$ we perform a grid search in a reasonable range 
(e.g., between the 10th and 90th percentiles of the~data).
In all forecasting experiments, the AR order and the quantiser thresholds 
are selected using the above procedure at the end of the training set.

\subsection{Alternative AR mixture models} 
\label{mar_comp}

\hspace{0.18in}
{\em Threshold AR models.}  
Threshold autoregressive (TAR) models~\cite{tong:80} have been 
used extensively in the analysis of nonlinear time series; see, 
e.g.,~\cite{tong:11,hansen:11} and 
the texts~\cite{cryer:book,tong:book}. 

{\color{black}
The most commonly used version of TAR models
is the self-exciting TAR (SETAR) model, which considers 
partitions of the state space based on the quantised value of $x_{n-d}$:
\begin{align}\label{setar}
\begin{split}
x_n =   \ \phi _ {0} ^ {(j)} +  \phi _ {1} ^ {(j)} \ x_{n-1} + \dots + \phi _ {p}^ {(j)} \ x_{n-p} + \sigma ^ {(j)}\ e_n , 
  \quad   \text{if} \;  Q(x_{n-d}) = j \in A   , 
\end{split}
\end{align}
where $ e_n \sim \mathcal N (0, 1)$, $p$ is the autoregressive order,  \mbox{$Q : \mathbb {R} \to A = \{0,\dots ,  m-1 \}$} is an $m$-ary quantiser of the form in (\ref{quant}), and $d$ is called the \textit{delay} parameter. In other words, the SETAR model class considers partitions of the state space based (only) on the value of $x_{n-d}$, with different parameters $(\boldsymbol \phi ^ {(j)}, \sigma ^ {(j)} )$ associated to each region.

It is clear from the above description that the BCT-AR model 
class is a strict generalisation of SETAR.
In specific, the BCT-AR model class always contains the SETAR models as specific instances, corresponding to particular trees in $\clT (D)$;
for example, any threshold model with delay parameter $d=1$ 
can be represented as a BCT-AR model with respect
to the full tree of depth $d=1$.
However, the BCT-AR class also contains other, more complicated tree-based 
partitions of the state space that cannot be represented as simple threshold models. For example,
asymmetric BCT-AR models with 
leaves at various depths define more complicated partitions 
that cannot be represented as linear combinations 
of threshold models. 
In fact, the BCT-X methodology can be viewed as a natural 
conceptual extension of the family of threshold models, which allows for a more systematic and powerful Bayesian way of breaking up the state space in possibly more -- and more complex -- regions, and then fitting a different time series model to each one of these regions.
}

\smallskip

{\em Mixture AR models.} The mixture autoregressive (MAR) models 
of~\cite{wong:00} consist of a simple 
linear mixture of~$K$ Gaussian AR components. 
When the BCT-AR posterior distribution essentially concentrates on $K$ models,
$T_1, \dots , T_K$, (which is both theoretically
`allowed' and is commonly observed in practice), 
the posterior predictive distribution can be expressed~as,  
\vspace*{-0.05cm}
\[
p(x_{n+1}|x) = \sum _{k=1} ^ K  \pi (T_k|x) \  p(x_{n+1}|T_k, x)   . 
\vspace*{-0.05cm}
\]
In this sense, BCT-AR can be viewed as a generalised MAR model, 
with components corresponding to the AR models at the leaves 
of each $T_k$, and Bayesian weights given by~$\pi (T_k|x) $. 
Therefore, the BCT-AR model class is a strict
generalisation of both 
MAR and SETAR. 

\smallskip

{\em Markov switching AR models.} The Markov switching autoregressive model 
(MSA)~\cite{hamilton:89} is a simple HMM, where the hidden state process is a 
discrete-valued first-order Markov chain (usually binary or ternary), 
and a different AR model is associated to each state. 
As the discrete states here are not observable~\cite{tsay:05},
the main difficulty 
in using the MSA is in estimating the model; 
the EM algorithm is usually employed 
for this task. This difficulty was also observed in practice 
in our experiments, 
where the MSA gave much larger running times compared to all other classical 
methods that do not involve neural networks 
(Supplemental Material, Section~\ref{run_times}).

\newpage

\section{BCT-AR: Experimental results} 
\label{experiments}

\vspace*{-0.05 cm}

{\color{black}
In this section, the performance of the BCT-AR model 
is evaluated on simulated and real-world datasets 
from standard applications of nonlinear time series 
in economics and finance. The complete descriptions of 
all datasets can be found in Section~\ref{list_of_data} 
of the Supplemental Material.
In all forecasting experiments, the training set consists 
of the first~50\% of the observations in each dataset, and we consider out-of-sample 1-step ahead forecasts, allowing every model to be updated at every timestep.
For the BCT-AR model, the current MAP tree model with its MAP estimated parameters 
is used at every timestep for forecasting.

Since,
in our experiments, the posterior $\pi(T|x)$ on model space
was typically found to concentrate on the MAP tree model with high posterior probability, 
this simple prediction rule already gives a reasonably good -- and computationally efficient -- approximation to 
the full Bayesian predictive distribution. In general,
however, it is noted that prediction based on model-averaging over trees $T$ 
and over parameters $\theta$, may also be easily implemented.

The BCT-AR model is compared with the most successful 
previous methods for these applications,
considering both classical and modern ML methods. 
Useful resources 
include the \texttt{R} package \texttt{forecast}~\cite{hyndman:08} and the 
Python 
library \mbox{`GluonTS'}~\cite{alexandrov:20}, 
containing implementations of state-of-the-art classical and ML methods, 
respectively. We briefly discuss the methods used, and refer to 
the packages' documentation and 
Section~\ref{train_details} of the Supplemental Material for more details on the methods and the training procedures carried out.
Among classical statistical approaches, we compare with 
ARIMA
and Exponential smoothing state space (ETS) models
(implemented 
in \texttt{forecast}), with
SETAR  -- using the conditional least squares (CLS) method implemented in the \texttt{R}~package 
\texttt{TSA}~\cite{tsa} -- and with MAR and MSA models, using the \texttt{R}~packages 
\texttt{mixAR}~\cite{mixar} and \texttt{MSwM}~\cite{mswm}.
Among ML-based techniques, we compare with 
the Neural Network AR (NNAR) model (implemented in \texttt{forecast}), and
with the most-successful RNN-based approaches, deepAR~\cite{salinas:20} and 
N-BEATS~\cite{oreshkin:19} -- both implemented in GluonTS.
}

\subsection{Simulated data} \label{sim}

\vspace*{-0.05 cm}

Here we use simulated data to 
illustrate that the BCT-X methods are consistent and 
effective with data actually generated by BCT-X models. 
A time series $x$ is simulated from the BCT-AR model with the
context-tree model $T$ in Figure~\ref{tree}, 
a binary quantiser with threshold $c=0$, 
and the following AR(2) base models at the three 
states defined 
by $T$:
\begin{align*} 
x_n \!=\! \left\{
\begin{array}{ll}
 0.7  \ x_{n-1} - 0.3  \ x_{n-2} + e_n,  \quad  e_n \sim \mathcal N (0, 0.15), \; \;  & \mbox{if} \ s = 1\text{:} \ \ \ x_{n-1}\geq 0, \\
-0.3  \ x_{n-1} - 0.2  \ x_{n-2} + e_n,  \quad  e_n \sim \mathcal N (0, 0.10),  \; \;    &\mbox{if} \ s = 01\text{:} \   x_{n-1}<0, \  x_{n-2}\geq 0, \\
 0.5 \ x_{n-1} + e_n,  \quad \quad \quad \quad \quad \quad e_n \sim \mathcal N (0, 0.05),  \; \; & \mbox{if} \ s = 00\text{:} \ x_{n-1}<0, \  x_{n-2}< 0.
\end{array}
\right.  
\end{align*} 

We first examine the posterior 
distribution $\pi(T|x)$ over $T\in\clT(D)$. With $n=100$ 
observations, the MAP context-tree model
identified by the GBCT algorithm is the `empty'
tree corresponding to a single AR model, with posterior probability~99.9\%. 
This means that the data do not provide sufficient evidence to support 
a more complex structure with multiple states. With $n=300$ observations, 
the MAP tree model is now 
the true underlying model, with posterior probability 57\%.
And with $n=500$ observations, the posterior of the true~model 
is~99.9\%. 
The complete BCT-AR model fitted from $n=1000$ observations with its MAP 
parameter estimates is shown in~(\ref{bctar_fitted}).
In Section~\ref{sim1} of the Supplemental Material
we~also report values of the 
evidence $p(x|c,p)$, which is maximised at the true values of $c=0$ and $p=2$.
\begin{align} \label{bctar_fitted}
x_n \!=\! \left\{
\begin{array}{ll}
 0.66  \ x_{n-1} - 0.19  \ x_{n-2} + e_n,  \quad  e_n \sim \mathcal N (0, 0.16), \; \;  & \mbox{if} \  x_{n-1}\geq 0, \\
-0.39  \ x_{n-1} - 0.27  \ x_{n-2} + e_n,  \quad  e_n \sim \mathcal N (0, 0.12),  \; \;    &\mbox{if}  \  x_{n-1}<0, \  x_{n-2}\geq 0, \\
 0.45 \ x_{n-1} - 0.03  \ x_{n-2}+ e_n,   \quad e_n \sim \mathcal N (0, 0.058),  \; \; & \mbox{if} \ x_{n-1}< 0, \  x_{n-2}\geq 0.
\end{array}
\right.  
\end{align}

\newpage

\noindent
{\bf Interpretation.}
Clearly, the three relevant states, the optimal quantiser,
and the correct AR order were all
identified without any prior training, based on only $n=1000$ observations.
Also, the posterior distribution of the AR parameters appears to
be well-concentrated around the true underlying values, verifying that all our inferential procedures are effective.

\medskip

\noindent
{\bf Forecasting.} 
The performance of the BCT-AR methods 
is evaluated in the task of out-of-sample 1-step ahead forecasts, 
and it is compared with state-of-the-art approaches in  
three simulated and three real datasets. 
The first simulated dataset~(\texttt{sim\_1}) consists of $n=600$ 
observations generated from the BCT-AR model used above,
the second~(\texttt{sim\_2}) has $n=500$ observations 
generated by a BCT-AR model with a ternary context-tree 
model of depth~2,
and the third~(\texttt{sim\_3}) consists of $n=200$ observations 
generated from a SETAR model of order $p=5$; see Section~\ref{list_of_data} 
of the Supplemental Material.

\begin{table}[!h]
  \centering
  \caption{Mean squared error (MSE) of forecasts with
simulated and real-world data.}
\label{t1}
\vspace*{-0.1 cm}
  \begin{tabular}{lccccccccccc}
\midrule
 & BCT-AR \hspace* { -0.2 cm} & ARIMA  \hspace* { -0.2 cm} & ETS \hspace* { -0.2 cm}  & NNAR \hspace* { -0.2 cm}  & DeepAR \hspace* { -0.2 cm}  & N-BEATS  \hspace* { -0.2 cm}  & MSA \hspace* { -0.2 cm}  & SETAR \hspace* { -0.2 cm}  & MAR \hspace* { -0.2 cm}  \\
 \midrule
 \texttt {sim\_1} \hspace* { -0.2 cm} &  \bf 0.131 & 0.150 & 0.178 & 0.143 & 0.148 & 0.232 & 0.142 & 0.141 & 0.151 \\
  \texttt {sim\_2} \hspace* { -0.2 cm} & \bf  0.035 & 0.050 & 0.054 & 0.048 & 0.061 & 0.112 & 0.049 & 0.050 & 0.064\\
  \texttt {sim\_3} \hspace* { -0.2 cm} & \bf 0.891  & 1.556 & 1.614 & 1.287  & 1.573 & 2.081 & 1.495 & 0.951 & 1.543 \\
  \midrule
  \texttt {unemp} \hspace* { -0.2 cm} & \bf 0.034 & 0.040 & 0.042 & 0.036 & 0.036 & 0.054 & \bf 0.034 & 0.038 & 0.037 \\
  \texttt {gnp} \hspace* { -0.2 cm} & \bf 0.324 & 0.364 & 0.378 & 0.393 & 0.473 & 0.490 & 0.353 & 0.394 & 0.384 \\
  \texttt{ibm} \hspace* { -0.2 cm} & 78.02 & 82.90 & 77.52 & 78.90 & \bf 75.71 & 77.90 & 81.68 & 81.07 & 77.02 \\
    \midrule
  \end{tabular}
\end{table}

The results in Table~\ref{t1} indicate that the BCT-AR model
outperforms all the alternatives in all three simulated experiments, 
achieving a mean-squared error (MSE) that is lower by between 7\% and 37\% 
compared to the second-best method. 
As discussed in Section~\ref{compl}
(see also Section~\ref{run_times} of the Supplemental Material
for more detailed results),
the BCT-X methods also outperform the alternatives in terms 
of empirical running times, by anywhere between one and three
orders of magnitude.
Next, we examine the performance of the BCT-AR methods 
in real-world applications from economics and finance.

\vspace*{-0.2 cm}

\subsection{US unemployment rate}

\vspace*{-0.1 cm}

An important application of SETAR models is in modelling the US 
unemployment 
rate~\cite{hansen:11,montgomery:98,tsay:05,rothman:98,koop:99}. 
As described in~\cite{montgomery:98,tsay:05}, the unemployment rate moves 
countercyclically with business cycles, and rises quickly but decays slowly, 
indicating nonlinear behaviour. Here, we examine the quarterly US unemployment 
rate in the period from 1948 to 2019 (dataset \texttt{unemp}, 
288 observations). Following~\cite{montgomery:98}, we consider the 
difference series \mbox{${\Delta x}_n = x_n-x_{n-1}$}, and also include 
a constant term in the AR model. For the quantiser alphabet size, 
$m=2$ is a natural choice here, as will become apparent below. The threshold 
selected using the procedure of Section~\ref{hyp} is $c=0.15$,
and the resulting MAP context-tree model
is the tree of Figure~\ref{tree}, with states 
$\mathcal {S} = \{ 1, 01, 00\}$ and posterior
probability~91.5\%. The complete BCT-AR model with its MAP parameters is given below, where $e_n \sim \mathcal N (0, 1)$,  
\[ 
{\Delta x}_n \!=\! \left\{
\begin{array}{ll}
  0.09 +  0.72  \ {\Delta x}_{n-1} - 0.30  \ {\Delta x}_{n-2} + 0.42 \ e_n,   \; \;  & \mbox{if} \  {\Delta x}_{n-1}\geq 0.15, \\
 0.04 + 0.29  \ {\Delta x}_{n-1} - 0.32  \ {\Delta x}_{n-2} + 0.32 \ e_n,  \; \;    &\mbox{if}  \  {\Delta x}_{n-1}< 0.15, \  {\Delta x}_{n-2}\geq 0.15, \\
 -0.02 + 0.34 \ {\Delta x}_{n-1} + 0.19  \ {\Delta x}_{n-2}+ 0.20 \ e_n,    \; \; & \mbox{if} \ {\Delta x}_{n-1}< 0.15, \  {\Delta x}_{n-2}< 0.15 .
\end{array}
\right.   
\]

\smallskip

\noindent
\textbf {Interpretation}. The MAP BCT-AR model finds significant structure in the data, providing a natural interpretation. It identifies 3 meaningful states: First, jumps in the unemployment rate higher than 0.15 signify economic contractions ($s=1$). If there is not a jump at the most recent time-point, the model looks further back to determine the state. 
The state $s=00$ signifies a stable economy, as there are no jumps in the unemployment rate for two consecutive quarters. Finally, $s=01$ identifies an intermediate state: ``stabilising just after a contraction''. An important feature identified by the BCT-AR model is that the volatility is different in each case. It is higher in contractions ($\sigma = 0.42$), smaller in stable economy regions ($\sigma = 0.20$), and 
in-between for state~01 ($\sigma = 0.32$). This is an important finding, verifying 
the economists' understanding that there is much higher uncertainty during
economic contractions. 

\medskip

\noindent
{\bf Forecasting.} In addition to its appealing interpretation, the 
BCT-AR model outperforms all benchmarks in forecasting (together with
MSA), achieving the lowest MSE (Table~\ref{t1}). In~terms of empirical running 
times, the BCT-AR model vastly outperforms MSA;
see  Section~\ref{run_times} of the Supplemental Material.

\subsection{US Gross National Product}
\label{s:USGNP}

 \begin{wrapfigure}{r}{0.33\linewidth}
  \begin{center}
  \vspace*{-0.85 cm}
    \includegraphics[width= 0.75 \linewidth, height= 0.4 \linewidth ]{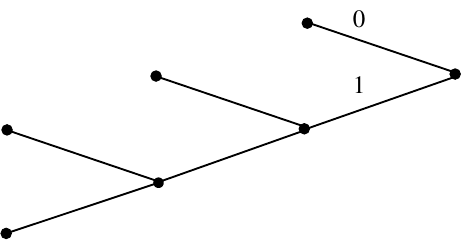}
  \end{center}
\vspace*{-0.5 cm}
\caption{ \color{black} MAP context-tree model for the GNP dataset.}
\label{tree_gnp}
\vspace*{-0.2 cm}
\end{wrapfigure}

Another important example of nonlinear time series in economics is the US 
Gross National Product~(GNP)~\cite{potter:95,hansen:11}. 
We examine the quarterly US GNP in the period from 1947 to~2019 
(dataset \texttt{gnp}, 291 observations). Following~\cite{potter:95}, 
here we consider the difference in the logarithm 
of the series, $y_n = \log x_n - \log x_{n-1}$. 
As above, $m=2$ is a natural choice for the quantiser size, helping to 
differentiate economic expansions from contractions, which govern the 
underlying dynamics. The MAP BCT-AR context-tree model is shown in 
Figure~\ref{tree_gnp}. It has maximum depth~3,
the states are $\mathcal {S} = \{0, 10, 110, 111 \}$,
and its posterior probability is~42.6\%. 
The complete set of MAP parameters at the leaves are shown in 
Section~\ref{data_gnp} of the Supplemental Material.

\medskip

\noindent
{\bf Interpretation.} Compared with the previous example, the 
MAP BCT-AR model identifies even richer structure in this dataset,
with four meaningful states. First, as before, there is a single 
state corresponding to an economic contraction -- which now corresponds 
to $s=0$ instead of $s=1$, as the GNP increases in expansions 
and decreases in contractions. And again, the model does not 
look further back whenever a contraction is detected -- it is 
a {\em renewal event}. In this example, the model shows that the effect of a contraction is still present even after {\em three} quarters ($s=110$),
and that the exact `distance' from a contraction is also important, with
the dynamics changing depending on how much time has elapsed.
Finally, the state $s=111$ corresponds to a flourishing, expanding economy.
An important feature 
captured by the model is again that the volatility is different in each 
case and, enhancing previous findings, it  is found to 
decrease with the 
distance from the last contraction. It starts 
at $\sigma =1.23$ for $s=0$ and decreases to $\sigma = 0.75$ 
for $s=111$ (see Section~{\ref{data_gnp}} of the Supplemental Material).

\medskip

\noindent
{\bf Forecasting.} 
As shown in Table~\ref{t1}, the BCT-AR model is found to outperform all benchmarks in forecasting in this dataset, likely due to 
the additional structure identified, giving
a significantly lower MSE than the second-best method
(by 9\%).

\begin{wrapfigure}{r}{0.3\linewidth}
  \begin{center}
  \vspace*{-2.1 cm}
    \includegraphics[width= 0.78 \linewidth, height= 0.4 \linewidth ]{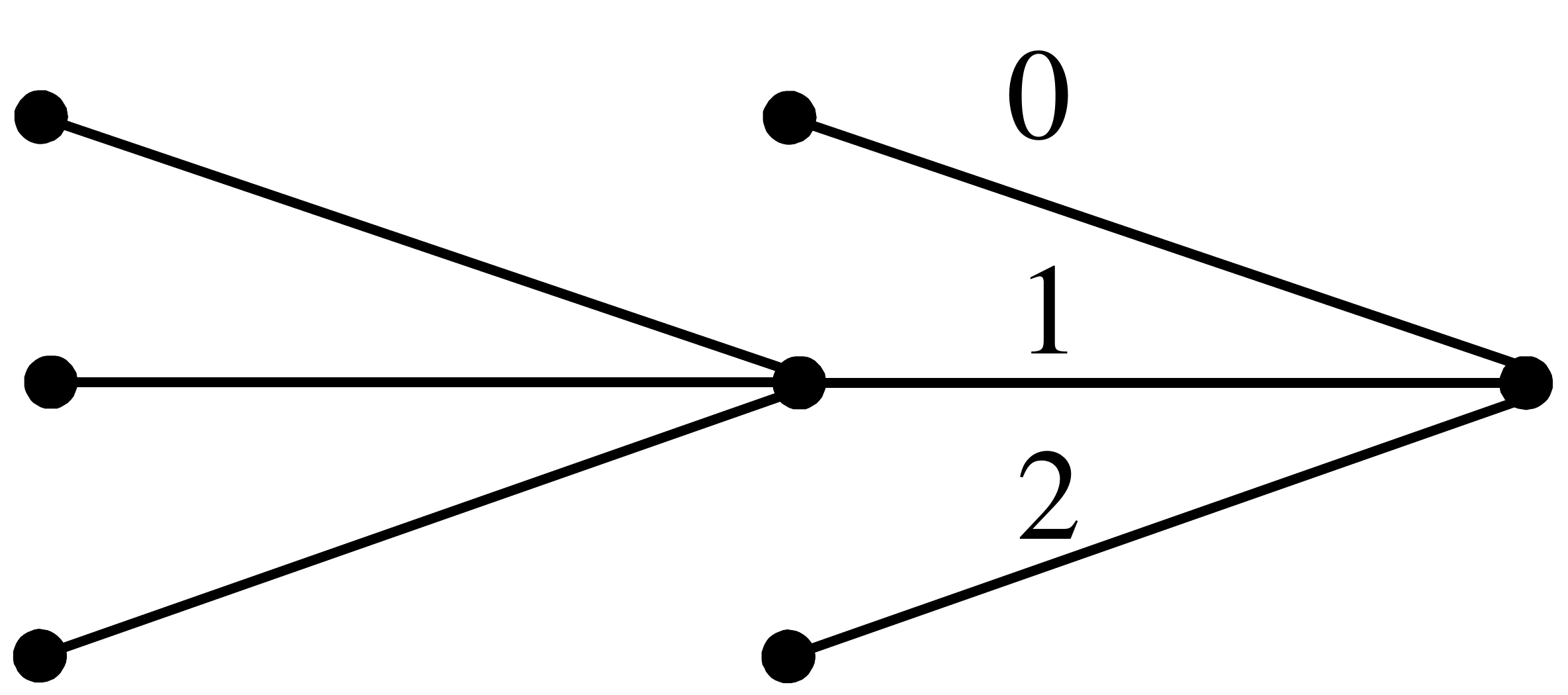}
  \end{center}
\vspace*{-0.5 cm}
\caption{\color{black} MAP context-tree model for the IBM dataset.}
\label{tree_ibm}
\vspace*{-0.3 cm}
\end{wrapfigure}

\subsection{IBM stock price}

We revisit the daily IBM common stock closing price~from~May~17, 1961 to  November~2, 1962 
(dataset \texttt{ibm}, 369 observations). 
This is a well-studied dataset, 
with~\cite{box:book} fitting an ARIMA model,
\cite{tong:book}~fitting a SETAR model, and
\cite{wong:00}~fitting a MAR model to the data. 
Following previous approaches, we consider the first-difference series, 
\mbox{${\Delta x}_n = x_n-x_{n-1}$}. For the alphabet size of the 
quantiser we choose $m=3$, with the values~$\{ 0,1,2\}$ naturally 
corresponding to the stock price changes \{down, steady, up\}.
Using the procedure of Section~\ref{hyp},
the resulting quantiser thresholds are $\{-7,+7\}$
and the AR order $p=2$.
The MAP context-tree model is shown in Figure~\ref{tree_ibm}. It has 
maximum depth 
equal to~2,
and it identifies five states,~$\mathcal {S} = \{0,2,10,11,12 \}$. Its posterior is 99.3\%, suggesting that there is very strong evidence in the data supporting this structure, even with only 369 observations. The complete BCT-AR model with its MAP parameters is given 
the Supplemental Material~(Section~\ref{data_ibm}).

\medskip

\noindent
{\bf Interpretation.} The BCT-AR model reveals important information about apparent structure in the data, which has not been identified before. Firstly, it admits a simple and natural interpretation: In order to determine the AR model generating the next value, one needs to look back until there is a significant enough price change (corresponding to the states 0, 2, 10,~12), or until reach the 
maximum depth of~2 is reached (state~11). Another important feature captured by this model is the commonly observed asymmetric response in volatility due to positive and negative shocks, sometimes called 
the {\em leverage effect}~\cite{tsay:05,box:book}. 
Indeed, the MAP model shows that negative shocks increase the volatility 
much more: the state $s=0$ has the highest volatility ($\sigma  =12.3$), 
with $s=10$ being a close second ($\sigma = 10.8$), showing that the effect 
of a past shock is still present. In all other cases the volatility 
is much smaller (between $\sigma = 5.17$ and $\sigma = 6.86$). 
Finally, we observe that when stabilising after a shock (states 10,~12), 
the latest value $x_{n-1}$ is not as important as~$x_{n-2}$, 
whereas $x_{n-1}$ is dominant in all other~cases.

\medskip

\noindent{\bf Forecasting.} From Table~\ref{t1}, it is observed that DeepAR outperforms all methods here, with
BCT-AR giving a marginally higher MSE but much smaller empirical running times;
see  Section~\ref{run_times} of the Supplemental Material.

\vspace*{-0,1 cm}

\section{BCT-ARCH: The Bayesian Context Trees Conditional Heteroscedastic Model} 
\label{bctgarch}

A key aspect of financial time series analysis is modelling 
the dynamic evolution of volatility over time. To capture the well-known 
{\em volatility clustering} present in financial data, 
Engle's seminal work introduced the autoregressive conditional 
heteroscedasticity~(ARCH) models~\cite{engle:82},
which, together with their numerous extensions, have been very 
widely used for modelling the volatility 
in financial time series. The perhaps more important limitation of 
ARCH models is their inability to describe the 
well-known asymmetric response in volatility due 
to positive and negative shocks -- 
the {\em leverage effect}~\cite{tsay:05,box:book}
mentioned in the previous section. A number 
of approaches have attempted to incorporate 
this feature in ARCH models, notably including the 
works of~\cite{nelson:91,glosten:93,zakoian:94,gray:96}; see 
Section~\ref{garch_comp} for more details.

In this section, as a second example of the general BCT-X class, 
we associate a different ARCH model to each state of the context-tree model, 
and refer to the resulting mixture model class as 
the {\em Bayesian Context Trees ARCH} 
(BCT-ARCH) model. The BCT-ARCH model is of high practical 
interest as 
it offers a systematic and powerful way of modelling  
volatility asymmetries.
As shown in Section~\ref{exp_arch},
it is found to outperform 
previous approaches in real examples.
Moreover, it reveals important structure present in the data 
that not been identified before, in the form of an 
{\em enhanced leverage effect}.

\vspace*{-0,1 cm}

\subsection{Bayesian modelling}

Given a quantiser $Q:\RL\to A=\{0,1,\ldots,m-1\}$, a fixed ARCH order $p$,
and a maximum context
depth $D\geq 0$, the BCT-ARCH model is the version of 
BCT-X with ARCH$(p)$ models as base modes.
With each potential state $s$ of
a context-tree model $T\in\clT(D)$
we associate an ARCH$(p)$ model,
\begin{align}
    x_n \sim {\mathcal N} \left (0, \sigma _ n ^ 2 \right ) , \; \; \; \; 
    \sigma _ n ^ 2  = \alpha _ {s,0} + \alpha _{s,1} x_{n-1} ^ 2 + \dots + \alpha _{s,p}x_{n-p} ^ 2 = \boldsymbol \alpha _ s ^{\text{T}} \ \mathbf{ { z } } _{n-1} ,
\end{align}
where $\boldsymbol \alpha _ s = \left (  \alpha _ {s,0} , \alpha _{s,1}  , \dots , \alpha _{s,p} \right )^{\text{T}}$ and $\mathbf{ { z } } _{n-1} = (1, x_{n-1} ^ 2,\dots, x_{n-p} ^ 2 )^{\text{T}}.$ 

\smallskip

Here, the parameters corresponding to each state $s$ 
are the ARCH coefficients, 
$\theta _ s = \boldsymbol \alpha _ s $. 
We follow~\cite{vrontos:00} and use the following non-informative 
priors: $\pi ( \alpha _ {s,0}) = 1 / \alpha _ {s,0}$, 
and $\pi ( \alpha _ {s,j})\sim U(0,1)$ for $1\leq j\leq p$.
Note that no hyperparameters need to be selected for this prior.

\subsection {Estimated probabilities and inference}
\label{s:Pest}

Our main tools for inference --
the 
GCTW and GBCT 
algorithms 
-- rely on the evaluation of the
estimated probabilities $P_e(s,x)$ 
given in~(\ref{pes}). However, in this case these are not available
in closed form. In fact, exact Bayesian inference is not possible even 
with a single ARCH model, and
MCMC approaches are typically
adopted for this task; see~\cite{virbickaite:15} for a review.
Therefore, we will employ effective approximations for 
the $P_e(s,x)$ and use them in GCTW and GBCT;
the remaining methodology remains the same as before. 

As noted earlier, the fact that $P_e(s,x)$ is in the form of 
marginal likelihoods allows us to use standard methods to approximate 
them. Specifically, at each node $s$ of $T_{\rm MAX}$ we
use a variant of the Laplace method described in~\cite{kass:95}, 
which was found to be effective in the case of a multivariate GARCH model 
in~\cite{vrontos:03}.
The resulting approximation is:
\begin{equation} \label{laplace}
    \widehat {P_e} (s,x) = (2 \pi ) ^ {p+1/2} \  | \widehat {I} _ s | ^ {1/2} \ \exp \big \{ L _ s  ( \widehat {\theta} _ s  ) \big \} \  \pi  ( \widehat {\theta} _ s  ).
\end{equation}
Here, 
$ L _ s  (  {\theta} _ s  )$ is the log-likelihood of data with context $s$,
\begin{equation} 
     L _ s  (  {\theta} _ s  ) =  \sum _{ i \in B_s} \log p (x_i | x_{-D+1} ^ {i-1} , \theta _ s ) ,
\end{equation}
$\widehat {\theta} _ s$ is its maximiser, which can be computed iteratively using the Fisher scoring algorithm, 
\begin{equation} 
\label{fisher}
    \widehat {\theta} _ s ^ { \ (k)} = \widehat {\theta} _ s ^ { \ (k-1)} + \widehat {I} _ s ^ {-1} \ {\frac{\partial L_s} {\partial \theta _s}  \bigg | _ {\theta _ s  = \  \widehat {\theta} _ s ^ { \ (k-1)}}} \; \; \; \; ,
\end{equation}
and $\widehat {I} _ s$ the expected information matrix  computed at $\widehat {\theta} _ s ^ { \ (k-1)}$. These quantities for the BCT-ARCH model are 
given in Lemma~\ref{lem:ARCH}, that is proven in Section~\ref{appB3} of the Supplemental
Material.

\begin{lemma}
\label{lem:ARCH}
For the BCT-ARCH model, the terms 
$L_s (\theta _ s), \partial L _s / \partial \theta _ s$,  and $ \widehat {I} _ s $ are given by,
\begin{align} \label{bctarch_logl}
& L_ s ( \theta _ s ) = -  \frac {| B _ s |} {2}  \log (2 \pi)-  \frac {1}{2} \sum _ {i \in B _s} \left ( \log \sigma _ i ^ 2 + \frac {x _ i ^ 2}  {\sigma _i ^ 2} \right ) , \\
& \frac{\partial L_s} {\partial \theta _s} =   \frac {1}{2} \sum _ {i \in B _s}  \frac {1}  {\sigma _i ^ 2} \left (\frac {x _ i ^ 2}  {\sigma _i ^ 2} -1 \right )  \mathbf{ { z } } _{i-1} , \label{bctarch_dl} \\
 & \widehat {I} _ s =  \left \{  - \mathbb {E} \left (  \frac {\partial ^ 2 L _s}{ \partial \theta _ s ^ 2 }\right ) \right  \} =  \frac {1}{2} \sum _ {i \in B _s}  \left ( \frac {1}  {\sigma _i ^ 4 } \right )  \mathbf{ { z } } _{i-1}   \mathbf{ { z } } _{i-1} ^ {  \text {T}}, \label{bctarch_mat}
\end{align}
where $B_s$ is as in {\em (\ref{lik})}, and $\sigma _ i ^ 2 = \theta _ s ^ {  \text {T}}  \ \mathbf{ { z } } _{i-1}$.
\end{lemma}

\textcolor{black}{
It is noted that another general alternative to the Laplace method for approximating the integrals of~(\ref{pes}) would be to employ MCMC 
as previously used in the literature to approximate marginal likelihoods; see, e.g.,~\cite{chib:95,chib:01}. 
}
Once the approximate values $\widehat {P_e} (s,x)$
of the estimated probabilities
$P_e (s,x)$ have been computed as outlined, inference
can be carried out exactly as in the case
of the BCT-AR model. In particular,
the GCTW, GBCT and $k$-GBCT algorithms can be used with these approximations for $P_e (s,x)$, and
the selection of the ARCH order $p$ and the 
quantiser thresholds $\{c_i\}$ can be performed
in the same way as before, using the method
described in Section~\ref{hyp} for the BCT-AR model.

\smallskip

In all the experiments presented in Section~\ref{exp_arch},
we take the number of iterations $M$ in~(\ref{fisher}) to be equal to~10;
see Section~\ref{app_sim_arch} of the Supplemental Material
for a discussion of this choice.

\subsection{Computational complexity}
\label{s:compl2}


Consider, as before, executing the GCTW or GBCT algorithm for a time series $x_1 ^n$. Assuming for the moment that
only one iteration of the Fisher scoring algorithm
is performed for each node $s$ of $T_{\text{MAX}}$,
the situation is the same with the BCT-AR case:
For each observation $x_i$, $1 \leq i \leq n$, exactly $D+1$ nodes need to be updated, 
and for each one of these nodes, only 
the quantities in (\ref{bctarch_logl})--(\ref{bctarch_mat}) need to be 
updated, which can be done efficiently by just adding one term to each 
sum. The only difference in the BCT-ARCH case is that this forward 
pass of the data for $1 \leq i \leq n$ now needs to be repeated for every iteration of 
the scoring algorithm, 
each time followed by the Fisher updates of (\ref{fisher}) for 
the nodes of $T_{\text{MAX}}$. Finally, the resulting estimates $\widehat {\theta} _ s$ can be used to compute the approximation $ \widehat {P_e} (s,x)$ of (\ref{laplace}) for all nodes $s$ of $T_{\text{MAX}}$, while the recursive step  remains identical with before. 

{
\color{black}
Arguing as in Section~\ref{compl}, and denoting the number 
of Fisher iterations as $M$, it is not hard to show that
the overall complexity 
of  the GCTW and  GBCT algorithms
for the BCT-ARCH model
is~$\mathcal {O} (nD (m+Mp^3))$.
Importantly, this is still linear in all $n$, $m$, and $D$, and is actually
only slightly higher than the $\mathcal {O} (nMp^2)$ complexity 
of fitting 
a single ARCH model, showing that, again, 
inference in this much richer model class can be 
performed at a negligible 
additional computational cost.
Also, it is easy to see that the memory requirements of the algorithms are again $\mathcal {O} (nDp^2)$, i.e., linear in both $n$ and~$D$.
}

\subsection{Alternative methods} 
\label{garch_comp}

In Section~\ref{exp_arch} we will compare the performance of the 
BCT-ARCH model with that of some of the most successful and 
commonly used alternative methods for modelling volatility.
A brief summary of these methods is given here.
Additional information can be found
in the \texttt{R}~packages \texttt{rugarch}~\cite{rugarch}, 
\texttt{MSGARCH}~\cite{gray:96}, and \texttt{stochvol}~\cite{HK:21},
and in Section~\ref{train_details} of the Supplemental Material.

\medskip

{\em GARCH and EGARCH models.} The most widely used extension of 
ARCH models are the Generalised ARCH (GARCH) models of~\cite{bollerslev:86}. 
The GARCH$(p,q)$ model is given by,
\begin{align}
    \sigma _ n ^ 2  = \alpha _ {0} + \sum _ {i=1} ^ p \alpha _{i} x_{n-i} ^ 2 + \sum _ {i=1} ^ q \beta _{i} \sigma _{n-i} ^ 2 ,
\end{align}
with the simple GARCH$(1,1)$ model being the most popular
in practice. An important extension of the GARCH model is the Exponential 
GARCH (EGARCH) model of~\cite{nelson:91}, where the parametrisation is
in terms of the logarithm of $\sigma _ n ^2$.

{\em GJR models.} A common alternative to GARCH,
aimed at explicitly capturing volatility asymmetries,
is the threshold model of Glosten, Jagannathan, and Runkle 
(GJR)~\cite{glosten:93}, given by,
\begin{align}
    \sigma _ n ^ 2  = \alpha _ {0} + \sum _ {i=1} ^ p 
\left ( \alpha _{i} + \gamma _ i \IND_{\{x_{n-i}<0\}}\right ) x_{n-i} ^ 2 
+ \sum _ {i=1} ^ q \beta _{i} \sigma _{n-i} ^ 2 ,
\end{align}
where the indicator function forces negative returns to have higher 
volatility. A similar 
model parametrised in terms of the standard deviation instead of the 
variance is given in~\cite{zakoian:94}.
\textcolor{black}{From the above description it is easy to see that
the GJR model can be viewed as a particular case of a threshold GARCH model.
So, the BCT-ARCH model class is more general than GJR,
since it allows for more complicated partitions of the state space, and it is hence expected to perform better 
in practice; see also the corresponding and more detailed discussion and comparison with threshold models that is carried oout for the BCT-AR model in Section~\ref{mar_comp}.}

\medskip

{\em MSGARCH models.} The structure of the Markov switching 
GARCH (MSGARCH)~\cite{gray:96,haas:04} model is identical
to that of MSA, except that
GARCH models replace the AR models associated to 
each discrete hidden state. As with the MSA, performing inference 
in this setting is much more challenging compared 
to the previous approaches. 

\medskip

{\em SV models.} An alternative to the ARCH family -- which 
models the evolution of volatility deterministically -- is the family 
of {\em stochastic volatility} (SV) models, which model the volatility 
as a random process. This is 
usually done using an HMM where the hidden state is the logarithm 
of the variance,  modelled 
as an AR process; see~\cite{shephard:09,HK:21} for more details. 
As with MSGARCH, the randomness 
present in the hidden state process makes inference much more 
challenging and computationally more expensive compared to 
the other approaches. 

\vspace*{-0.2 cm}

\section{BCT-ARCH: Experimental results} 
\label{exp_arch}

\vspace*{-0.15 cm}

\subsection{Simulated data}
\label{s:simulated}

We begin by examining the performance of the BCT-ARCH inference methods 
on data generated from a model within this class.
Consider a time series $x$ consisting of observations 
simulated from a BCT-ARCH model where the context-tree
model is the binary tree of depth~1 with states $\mathcal {S} = \{ 0,1 \}$,
and with the associated ARCH(2) base models:
\begin{equation} 
\sigma _n ^ 2 \!=\! \left\{
\begin{array}{ll}
 0.10 + 0.20  \ x_{n-1} ^ 2 +  0.20  \ x_{n-2} ^ 2,\;\;   
	& \mbox{if} \  x_{n-1}\leq 0, \\
 0.10 + 0.20  \ x_{n-1} ^ 2  ,\;\;
	& \mbox{if} \  x_{n-1}>0.
\end{array}
\right.
\label{eq:MAPs1}
\end{equation}
This is an intentionally difficult example, with two very similar ARCH models 
in the two~regions.
Let the maximum context depth $D=5$ and consider a binary quantiser
with threshold $c$.
With $n=1000$ observations the MAP context-tree model is the `empty' tree, 
corresponding to a single ARCH model, as there are not enough data 
to reveal a more complex structure. With $n=2500$ observations, 
the MAP context-tree model is now the true underlying model 
with a posterior 
probability of 42\%, 
and with $n=5000$ observations the posterior probability
of the true model becomes 90\%.
The parameter estimates from $n=5000$ observations are,
\begin{equation} 
\widehat {\sigma} _n ^ 2 \!=\! \left\{
\begin{array}{ll}
 0.10 + 0.20  \ x_{n-1} ^ 2 +  0.16  \ x_{n-2} ^ 2, \; \;  
	& \mbox{if} \  x_{n-1}\leq 0, \\
 0.10 + 0.21  \ x_{n-1} ^ 2  +  0.02  \ x_{n-2} ^ 2 , \; \;  
	& \mbox{if} \  x_{n-1}>0.
\end{array}
\right. 
\label{eq:MAPs2}
\end{equation}

Here, the ARCH order $p=2$ and the quantiser threshold $c=0$
were selected as described in Section~\ref{s:Pest}.
More extensive results on simulated data are shown in 
Section~\ref{app_sim_arch} of the Supplemental Material,
illustrating that the posterior probability of the true underlying 
model converges to~1, and that the estimates of the parameters converge 
to the true parameter values. 

\subsection{Volatility in financial indices}
\label{s:volatility}

In this section, the BCT-ARCH framework is used to model 
the volatility on four real-world datasets corresponding to time
series describing four major financial indices. In each example, 
the $n=7821$ daily values of a stock market index are examined, 
during a period of thirty years ending on 7 April 2023. 
Similarly to Section~\ref{s:USGNP}, we examine the 
transformed time series, $y_n = 10  [\log x_ n -\log  x_{n-1}]$.
In order to explicitly capture the leverage effect and distinguish 
between positive and negative shocks, a binary quantiser
with threshold $c=0$ is used. The two ``memory length'' parameters,
namely, the maximum context depth
$D$ and the ARCH order $p$, are both taken to be equal to~5,
corresponding to a week of trading days.

\medskip

\noindent 
{\bf Results: New structure.} The most important feature captured by the 
BCT-ARCH model used on stock market index data
is an {\em enhanced leverage effect}. The fitted models suggest that,
in understanding the asymmetries present in volatility, 
it is not only the sign of the most recent value-change that is relevant, 
but the exact pattern of recent `ups' and `downs' is also important. 
These relevant patterns are 
identified from the data as the {\em states} given by the leaves
of the fitted MAP context-tree model.
In all four cases examined here,
that set of relevant states was found to be 
richer than $ \mathcal {S}=\{0,1\}$, 
which would correspond to just the sign of the previous change as in 
traditional leverage modelling. This strongly suggests the conclusion
that modelling such enhanced leverage is required in practice, 
and that the BCT-ARCH model reveals this essential structure that 
has not been identified before.

\medskip

\noindent
{\bf FTSE~100.} 
First we examine the Financial Times Stock Exchange 100 
Index, which is the most commonly used UK-based stock market indicator, 
including 100 companies listed on the London Stock Exchange. 
The fitted BCT-ARCH model exhibits the enhanced leverage effect, 
as it identifies the relevance of three meaningful states.

The MAP 
context-tree 
model given by the GBCT algorithm is shown in Figure~\ref{fig:bctarch}; 
it has depth~2 and three leaves, $\mathcal {S} = \{ 0, 10, 11\}$. 
Its posterior probability is~95.2\%, signifying that 
there is very strong evidence 
in the data supporting this exact structure. 
State $s=0$ corresponds to a negative shock at the last timestep, $s=10$ to stabilising just after a negative shock whose effect is still present, and $s=11$ to a flourishing period. 
The complete BCT-ARCH model is given by,  \vspace*{-0.1 cm}
\begin{align*}
     {\sigma} _n ^ 2 \!=\! \left\{
\begin{array}{ll}
 0.00 + 0.21 \ y_{n-1} ^ 2 +  0.16  \ y_{n-2} ^ 2 +  0.21  \ y_{n-3} ^ 2 +  0.18  \ y_{n-4} ^ 2 +  0.13  \ y_{n-5} ^ 2, \; \;  & \mbox{if} \  s=0, \\
 0.00 + 0.02 \ y_{n-1} ^ 2 +  0.19  \ y_{n-2} ^ 2 +  0.21  \ y_{n-3} ^ 2 +  0.15  \ y_{n-4} ^ 2 +  0.12  \ y_{n-5} ^ 2, \; \;  & \mbox{if} \  s=10, \\
 0.00 + 0.00 \ y_{n-1} ^ 2 +  0.10  \ y_{n-2} ^ 2 +  0.12  \ y_{n-3} ^ 2 +  0.09  \ y_{n-4} ^ 2 +  0.11  \ y_{n-5} ^ 2, \; \;  & \mbox{if} \  s=11.
\end{array}
\right.  \vspace*{-0.05 cm}
\end{align*}

Among the three states,
the ARCH coefficients of state $s=11$ (where no negative
shocks are present) are the smallest for all
five lags. 
States $s=0$ and $s=10$ have very similar coefficients except for 
the one corresponding to the first lag, $y _{n-1} ^2$, which is 
essentially zero for $s=10$, i.e., 
when it corresponds to an increase in value.
These observations are consistent 
with the common understanding that negative shocks lead 
to greater increases in volatility.

\begin{figure}[h!]
\vspace*{-0.1 cm}
\vspace*{-0.05 in}
\centering
\includegraphics[width= 0.25 \linewidth, height= 0.16 \linewidth ]
	{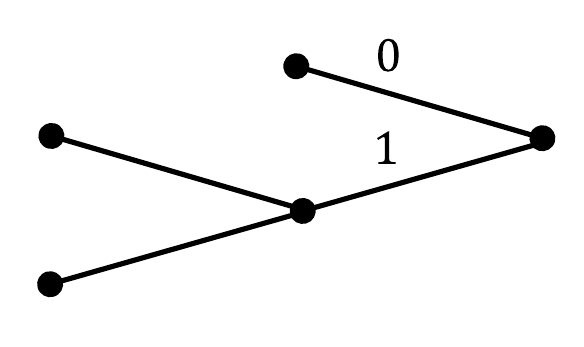}~
\hspace*{0.3 cm} 
\includegraphics[width= 0.33 \linewidth, height= 0.18 \linewidth ]
	{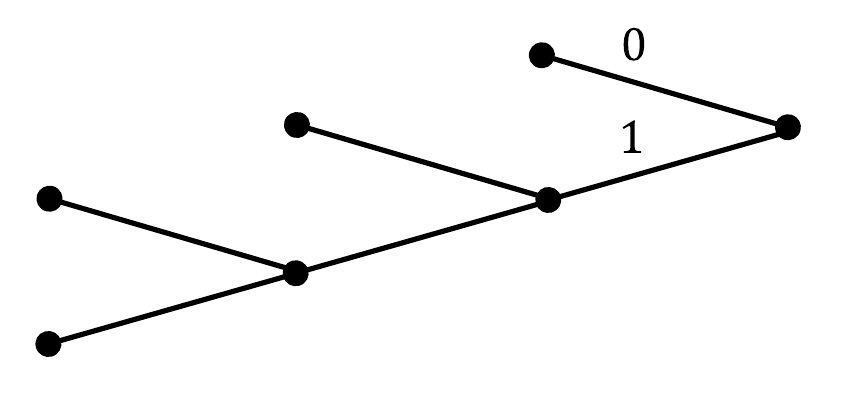}~
\includegraphics[width= 0.33 \linewidth, height= 0.2 \linewidth ]
	{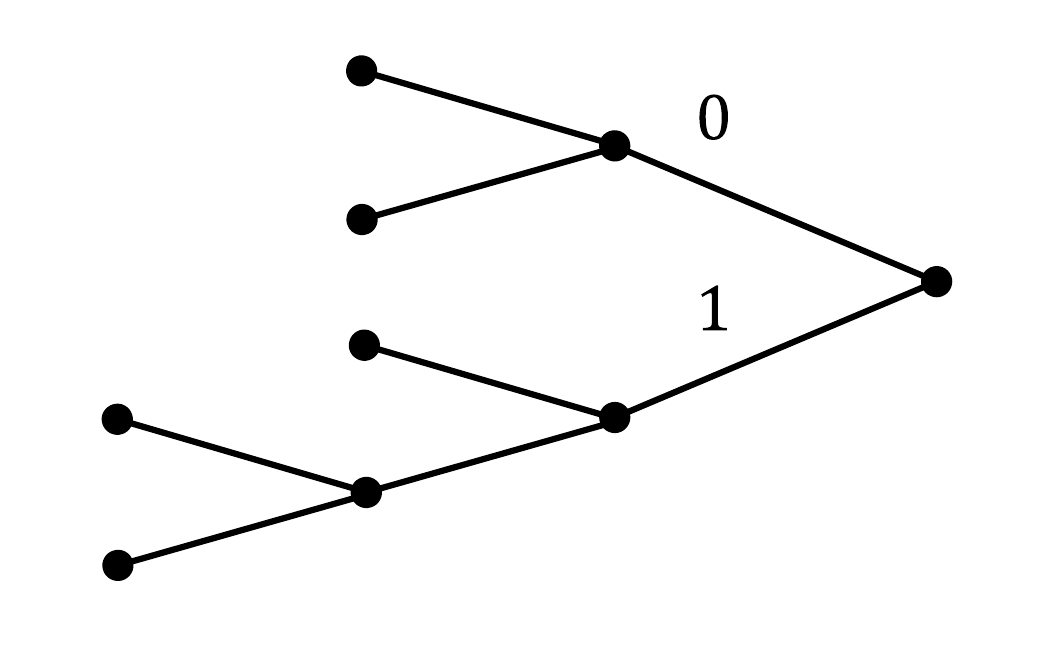}
\vspace*{-0.10in}
\caption{MAP context-tree models for major stock market indices: 
FTSE~100 (left), CAC~40 and DAX (middle), S\&P~500 (right).}
 \label{fig:bctarch}
 \vspace*{-0.10 in}
\end{figure}

\medskip

\noindent
{\bf CAC~40 and DAX.} Next we examine
CAC~40 and DAX, the main stock market indices in 
France and Germany, respectively, each one consisting of 40 major 
companies. The model fitted to 
the CAC and DAX time series finds
very similar structure present in the two datasets:
The same MAP context-tree model is identified by the GBCT 
algorithm in both cases~(Figure~\ref{fig:bctarch}). 
It has depth~$3$ and four leaves, 
$\mathcal {S} =  \{ 0, 10, 110, 111\} $. 
The estimated ARCH coefficients for the two models
are given 
in Sections~\ref{appc8} and~\ref{appc9} of the Supplemental Material.

The interpretation of this result is similar to that in the case 
of FTSE, with negative shocks again playing the
role of a `renewal' event:
In order to determine the distribution of the 
current state, the model looks back into the past 
until the first time a negative shock is detected. 
The only difference is in the 
memory of the discrete state process, with negative shocks 
now being relevant even 
if they occur {\em three} timesteps in the past. This gives a total 
of four different states, meaning that some additional structure has 
been identified. In each case, the ARCH coefficients corresponding
to lags with positive changes are small: 
For example, for $s=10$ the coefficient $\alpha _ 1 \approx 0$, 
while for $s=110$ both $\alpha _ 1 $ and~$\alpha _ 2$ are small. 
Overall, the BCT-ARCH model again exhibits the enhanced leverage effect,
and it gives strong evidence of a rich asymmetric 
response in volatility, with negative shocks having a stronger effect. 

\medskip

\noindent
{\bf S\&P~500.} The last index we examine
is Standard and Poor's 500 (S\&P~500), one of the most commonly 
followed indices worldwide, consisting of 500 of the largest companies 
in the~US. Here, the MAP context-tree model 
(Figure~\ref{fig:bctarch}) 
again displays the enhanced leverage effect, 
and it reveals even more structure
compared to the previous cases.  It has 
depth~$3$, five leaves, $\mathcal {S} = \{ 00, 01, 10, 110, 111 \}$, 
and a posterior probability 
of $91.4 \%$. In fact, the MAP context-tree model is the same tree 
as for CAC and DAX, but with one additional 
branch added at $s=0$. 
Apart from identifying slightly more structure, the 
interpretation of the BCT-ARCH model is very similar with before. 
The values of the estimated ARCH coefficients 
(reported in Section~\ref{appc10} of the Supplemental Material) are 
small when they correspond to lags with positive changes, 
again describing an asymmetric volatility response.

\vspace*{-0.2 cm}

\subsection{Forecasting performance} \label{arch_forecasting}

In this section, the performance of the BCT-ARCH model
in forecasting is illustrated and contrasted with that of 
the alternative methods 
outlined in Section~\ref{garch_comp}. 
As the volatility is not directly observable, effectively comparing 
different models is known to be a challenging~task~\cite{tsay:05}. 
Following standard approaches~\cite{geweke2010comparing,vrontos:03b,dellaportas:07}, 
in order to measure the relative predictive ability of the models, 
we consider the predictive distributions in 1-step ahead 
out-of-sample forecasting. 

Following~\cite{dellaportas:07}, the last 130 observations 
are taken as the test set in each of the four datasets,
corresponding to the trading days during a period
of six months in each case; the first 7691~observations are used 
as the training set.
\textcolor{black}{
At every timestep~$i$ in the test set,
all models are estimated using the entire past,
and the resulting predictive density 
$\widehat {p}  ( y _i | y_{1} ^ {i-1}  )$
is evaluated at the next test datapoint, $y_i$.
The complete details of the training process employed for each method 
are given in Section~\ref{train_details} of the Supplemental Material.
For the BCT-ARCH model, as in the case of the BCT-AR model
earlier, we use the MAP tree model with its MAP estimated parameters 
for forecasting. 
Also, it is noted that, throughout this section, the threshold of the 
quantiser is still fixed at $c=0$ in order to explicitly differentiate 
between positive and negative shocks; but in terms of forecasting 
performance, the results of BCT-ARCH could possibly be further improved 
by also estimating the threshold from data in the usual way.}

{\color{black}
As is standard practice~\cite{geweke2010comparing,vrontos:03b}, we examine 
the logarithm of the predictive density, i.e.,
$$ \vspace*{-0.2 cm} \mathcal {L} = \sum _ {i  } \log 
\widehat {p}  ( y _i | y_{1} ^ {i-1}  ), \vspace*{-0.0 cm} $$
evaluated over all datapoints $y_i$ in the test set, in  one-step ahead out-of-sample forecasts.  
}

\newpage

\begin{table}[!h]
  \centering
  \caption{ \color{black} Comparing the predictive ability of volatility models in terms
	of the log predictive density.}
  \vspace{-0.15 cm}
\label{table_arcH_real}
  \begin{tabular}{lccccccc}
\midrule
 & BCT-ARCH  & \; ARCH   &  \; GARCH  & \; GJR  & \; EGARCH   &  \;  MSGARCH  & \; \ SV  \\
 \midrule
 \texttt {ftse} & $ \bf 161.9$ & $157.7$ & $154.5$ & $159.7$ & $159.0 $& $159.7$ & $154.4$ \\
  \texttt {cac40} & $ \bf 112.5$ & $108.6$ & $108.7$ & $111.0$ & $ 112.4 $& $109.2$ & $106.9$ \\
  \texttt {dax} & $ \bf 111.7$ & $105.9$ & $105.4$ & $106.4$ & $107.5 $& $106.1$ & $103.2$  \\
  \texttt {s\&p} & $  78.73$ & $74.89$ & $81.04$ & $83.89$ & $ \bf 84.58 $& $80.95$ & $80.16$ \\
    \midrule
  \end{tabular}
  \vspace{-0.2 cm}
\end{table}

The results for the four stock market indices are presented 
in Table~\ref{table_arcH_real}. The BCT-ARCH model 
is seen to outperform all the alternatives in all examples, the only exception being  
the S\&P index data. Because of its low computational complexity 
and  efficient sequential updates, 
the BCT-ARCH model also outperforms all the alternatives in terms 
of its computational requirements, giving empirical running times 
that anywhere between one and three orders of magnitude smaller than those of the alternatives;
see Section~\ref{run_times} of the Supplemental Material.


\vspace*{-0.15 cm}

\subsection{Statistical significance tests}
\label{sec:stat_sig}

Finally, we further validate the findings of the previous section, namely that the BCT-ARCH model consistently outperforms the alternatives in forecasting the volatility of stock market indices  -- mainly because of its ability to model asymmetries in a more flexible and systematic way. In this section we repeat the above experiment for a number of important stock market indices and test for statistical significance of the results. 

\begin{table}[!h]
  \centering
  \caption{ \color{black} Comparing the forecasting performance of different volatility models in terms
	of the log predictive density, for major European stock market indices.}
  \vspace{-0.15 cm}
\label{table_stat_sig}
{  \begin{tabular}{lcccc}
\midrule
Index & BCT-ARCH   & \; GJR  & \; EGARCH   &  \;  MSGARCH    \\
\midrule
Austria: ATX & \bf  125.7 &  122.8 &  119.0 & 122.1 \\
Belgium: Bel-20 &  \bf 162.4  & 161.0  & 161.5 & 161.0 \\
Denmark: OMX Copenhagen &  \bf 28.80 &  18.96 & 21.22  & 21.02 \\
Europe Dow & \bf  128.5 &  125.1 &  123.8 &  127.3 \\
EURO STOXX 50 & \bf 104.4  & 99.27 & 102.5 & 102.6 \\

Finland: OMX Helsinki &  141.4 & \bf 146.2 & 144.5 & 145.6 \\

 France: CAC 40  & \bf 112.5 & 111.0 & 112.4 & 109.2 \\
 
 Germany: DAX & \bf 111.7  & 106.4  & 107.5 & 106.1 \\

 Greece: Athex Composite &  137.0 & 137.1 & \bf 137.5 & 125.9 \\

 Ireland: ISEQ All-Share &  \bf  123.0 & 118.3 &  117.3 & 119.8 \\

 Italy: FTSE MIB &  \bf 97.28 & 94.58  & 96.06  & 94.48 \\

 Netherlands: AEX  & \bf 111.4 & 107.3 &  110.4 & 110.2 \\

 Norway: OBX Index & 135.9 & 139.3 & 140.3 & \bf 140.9 \\

 Portugal: PSI 20 &  150.1 & 148.7 & 149.2 & \bf  151.6 \\

 Spain: IBEX 35 & \bf 129.0 & 125.5 & 125.3 &  122.4 \\

STOXX Europe 600 & \bf 132.3 &128.4 & 130.9  & 130.3 \\

Sweden: OMX Stockholm 30 &  \bf134.8 & 133.2 &  132.9 &  131.8 \\

Switzerland: SMI &  156.7 &  158.7 & 155.9 & \bf 158.9 \\

 UK: FTSE 100 & \bf 161.9  &159.7  & 159.0 & 159.7 \\
UK: FTSE All-Share & \bf  162.5 & 160.3 &  159.8 & 156.1 \\
 US: S\&P~500 &  78.73 & 83.89 & \bf 84.58 & 80.95 \\
 \midrule

 Average Rank & \bf 1.67 & 2.81 &  2.67 & 2.86 \\
    \midrule
  \end{tabular}
}
\vspace*{-0.15 cm}
\end{table}

\newpage

Specifically, apart from the S\&P~500 index we repeat the above experiment for 20 major European stock market indices (including the FTSE~100, CAC~40 and DAX as before); again, a complete set of the details for each dataset/index is given in Section~\ref{list_of_data} of the Supplemental Material. In this section, we exclude from our comparisons the simple ARCH, GARCH and SV models, since they are more restricted model classes that do not account for volatility asymmetries, and as a result they were all found to perform consistently worse than the other methods in the previous section.
{\color{black}
The log predictive density results for these 21 major stock market indices are presented in Table~\ref{table_stat_sig}.
}

The results of Table~\ref{table_stat_sig} clearly show that the BCT-ARCH model outperforms the alternatives in volatility forecasting: It achieves the best performance in 15 out of 21 datasets, and it also achieves the best average ranking overall, which is 1.67 compared to 2.67 for the second-best EGARCH model. Further, we examine
the statistical significance of these findings by implementing post-hoc Nemenyi tests~\cite{hollander2013nonparametric,demvsar2006statistical} using the \texttt{R} package \texttt{tsutils}~\cite{tsutils}; the results are presented  in Figure~\ref{nemenyi}.

\begin{figure}[h!]
    \centering
    \includegraphics[width=0.49\linewidth]{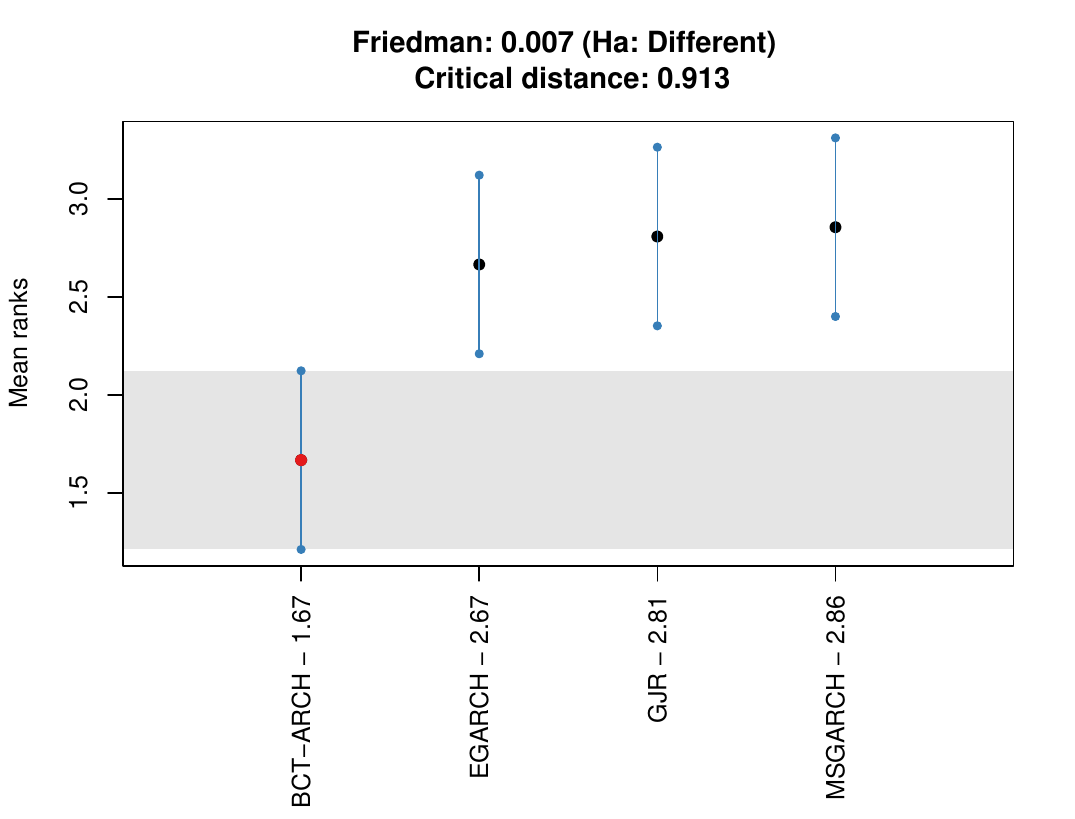}~
    \includegraphics[width=0.49\linewidth]{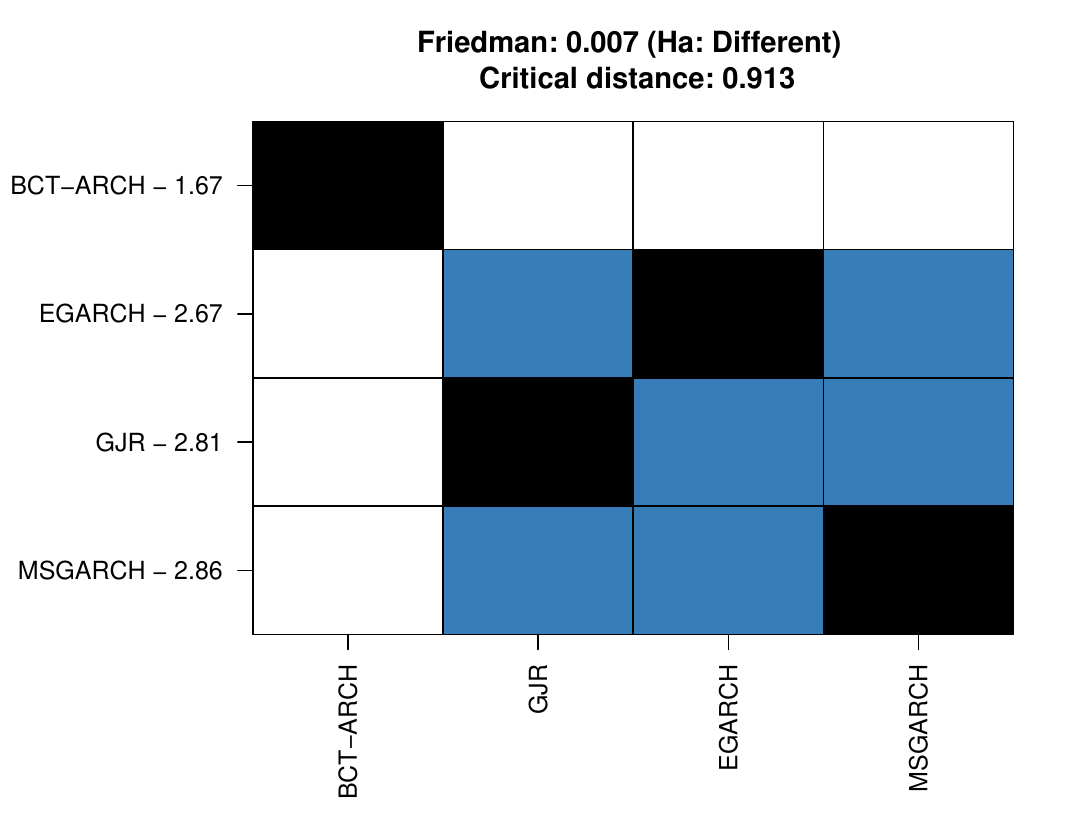}
    \caption{
    Post-hoc Nemenyi tests for volatility forecasting. Left: The MCB plot indicates statistically significant differences in performance when the distance in average ranking in greater than the critical distance of the test. Right: Matrix plot. White coloured cells signify a statistically significant difference in performance between the corresponding (row--column) methods, while blue cells signify the lack of statistically significant differences. }
    \label{nemenyi}
    \vspace*{-0.1 cm}
\end{figure}

The Friedman test~\cite{hollander2013nonparametric} rejects the null 
hypothesis that all methods perform similarly, and the post-hoc Nemenyi 
tests shown in Figure~\ref{nemenyi} further justify our findings. 
At the 90\% confidence level, there is a statistically significant 
difference between the performance of the BCT-ARCH model and that 
of any other method. In particular, the Multiple Comparison with Best (MCB) 
plot~\cite{koning2005m3} indicates that, in each case, the 
corresponding difference in average ranking is greater than the 
critical distance of the test. Moreover, the matrix 
plot~\cite{kourentzes2019cross} suggests that 
all other methods (EGARCH, GJR, MSGARCH) 
have comparable performance (among them),
as there is not sufficient evidence of
statistically significant differences
between any pair of them.

{\color{black}

 \section{Concluding remarks and future work} \label{conclusions}
%
This work develops a general Bayesian framework for building flexible 
and interpretable mixture models for real-valued time series,  
that are based on context trees. 
The proposed framework can be combined with any existing model class 
as a base model, resulting in a much richer class of flexible mixture models, 
for which we provide algorithms that allow for Bayesian inference at a negligible additional
computational cost compared to the original model. The utility of 
the proposed methodology has been illustrated by using AR and ARCH 
models as the base model, in both cases resulting in flexible 
mixture models of high practical interest.
 
The generality of the proposed framework leads to several possible 
directions for future work:
For any given application, BCT-X can be combined with any existing 
state-of-the-art model to provide much greater modelling flexibility
as well as potentially significant improvements in forecasting
performance. As a few examples, candidate base model classes include
ARIMA, EGARCH, general state space models, and MAR models, 
potentially leading to new and powerful ways to model 
feature-diverse time series datasets~\cite{kang2020gratis}. 
Similarly, BCT-X could be 
combined with modern ML models, including GPs, neural networks, 
and Deep Learning methods like DeepAR.
In a different direction, forecasting performance might be improved
by employing combination tools, ranging from simple averaging methods 
to modern ensemble learning techniques like bagging and 
boosting, which might lead to important practical improvements.
Closing, we remark that the entire BCT-X framework 
can be extended to the
multivariate time series setting, something that would greatly broaden the scope of potential applications.
 }

{\small

\bibliographystyle{plain}

\def\cprime{$'$}

}

\newpage

\appendix

\begin{center}
{\huge
{\bf Supplementary Material}\\
}
\end{center}

\section*{Data and code availability}

A reproducibility package which contains the datasets and code used in  this paper is available in the online repository: \url{https://github.com/IoannisPapageorgiou/Replication_BCTX}.

\section{Proofs of Theorems} 
\label{appA}

The key observation in the proofs is that, due to the form of the
estimated probabilities $P_e(s,x)$ in~(\ref{pes}),
it is
possible to factorise the marginal likelihoods $p(x|T)$ as,
\begin{equation*}
p(x|T)
=\int p(x|\theta,T)\pi(\theta|T)  d\theta =\int  \prod _{s \in T} \bigg (  \prod _ {i \in B_s} p(x_i| T, \theta _s , x_{-D+1}^{i-1} )  \ \pi (\theta_s) \  d\theta_s \bigg )  
=\prod_{s\in T} P_e(s,x),
\end{equation*}
where the second equality follows from the general BCT-X likelihood 
in~(\ref{lik}) and the fact that the priors on the 
parameters at the leaves are independent, so that
$\pi(\theta | T) = \prod _ {s  \in T }\pi (\theta _ s)$. 

Then, the proofs of Theorems~\ref{ctwth} and~\ref{bctth}
follow along the same lines as the proofs of the corresponding results for discrete time series in~\cite{BCT-JRSSB:22}.
The only difference is that the estimated probabilities $P_e(s,x)$ 
of~(\ref{pes}) are used in place of their simple discrete versions.
Before giving the proofs of the theorems, we recall a useful property 
for the BCT prior $\pi_D(T)$.
Let $\Lambda=\{\lambda\}$ denote the empty tree consisting 
only of the root node $\lambda$. 
Any tree $T\neq\Lambda$ can be expressed as the~union $T=\cup_j T_j$ of a collection of $m$ subtrees $T_0,T_1,\ldots,T_{m-1}$,
and its prior can be decomposed as~\cite{BCT-JRSSB:22}:

\begin{lemma}
\label{lem:union}
\ If $T\in{\mathcal T}(D)$, $T\neq\Lambda$, is
expressed as the union $T=\cup_jT_j$ of the
subtrees \mbox{$T_j\in{\mathcal T}(D-1)$},
then,
\begin{equation}
\pi_{D}(T) = \alpha^{m-1} \prod_{j=0}^{m-1} \pi_{D-1}(T_j).
\label{eq:prior-ind}
\end{equation}
\end{lemma}

\subsection{Proof of Theorem~\ref{ctwth}}

The proof is by induction. We want to show that: 
\begin{equation}
P_{w,\lambda} = p(x) = \sum_{T \in {\mathcal T}(D)}  \pi(T) p(x|T) =
\sum_{T \in {\mathcal T}(D)} 
\pi_{D}(T)\prod_{s \in T} P_{e} (s,x).
\label{eq:target}
\end{equation}
We claim that the following more general 
statement holds: For any node $s$ at depth
$d$ with $0\leq d\leq D$, we have,
\begin{equation}
P_{w,s} = \sum_{U \in {\mathcal T}(D-d)} \pi_{D-d}(U) 
\prod_{u \in U} P_{e} (su,x),
\label{eq:claim}
\end{equation}
where $su$ denotes the concatenation of contexts
$s$ and $u$.

\smallskip

Clearly~(\ref{eq:claim}) implies~(\ref{eq:target}) upon
taking $s=\lambda$ (i.e., with $d=0$). Also,~(\ref{eq:claim}) is 
trivially true for nodes
$s$ at level $D$, since it reduces to the 
fact that $P_{w,s}=P_{e,s}$ for leaves $s$,
by definition.

Suppose~(\ref{eq:claim}) holds for all nodes $s$
at depth $d$ for some fixed $0 < d \leq D$. 
Let $s$ be a node at depth $d-1$;
then, by the inductive hypothesis,
\begin{align*}
P_{w,s} 
=& \beta P_e(s,x) + (1-\beta)\prod_{j=0}^{m-1} P_{w,sj} \\
=& \beta P_e(s,x) + (1-\beta)\prod_{j=0}^{m-1}\left 
	[\sum_{T_j \in {\mathcal T}(D-d)} 
	\pi_{D-d}(T_j) \prod_{t \in T_j}P_e({sjt},x) \right],
\end{align*}
where $sjt$ denotes the concatenation of context
$s$, then symbol $j$, then context $t$, in that order. 
So,
\begin{align*}
P_{w,s} 
=& \beta P_e(s,x) + (1-\beta) \sum_{T_0,T_1,\ldots,T_{m-1}\in{\mathcal T}(D-d)}
	\prod_{j=0}^{m-1} \left[ \pi_{D-d}(T_{j})
	\prod_{t \in T_j}P_{e}({sjt},x) \right ] \\
=& \beta P_e(s,x) + \frac{1-\beta}{\alpha^{m-1}}
	\sum_{T_0,T_1,\ldots,T_{m-1}\in {\mathcal T}(D-d)} 
	\pi_{D-d+1}(\cup_jT_j) 
	\left[ 
	\prod_{j=0}^{m-1} 
	\prod_{t \in T_j} P_e({sjt},x) \right],
\end{align*}
where 
for the last step we have used~(\ref{eq:prior-ind})
from Lemma~\ref{lem:union}.

Concatenating every symbol $j$ with every leaf of the 
corresponding tree $T_j$, we end up with all the leaves
of the larger tree $\cup_jT_j$. Therefore,
\begin{align*}
P_{w,s} 
= \beta P_e(s,x) + \frac{1-\beta}{\alpha^{m-1}}
	\sum_{T_0,T_1,\ldots,T_{m-1}\in {\mathcal T}(D-d)} 
	\pi_{D-d+1}(\cup_jT_j) 
	\prod_{t \in \cup_jT_j} P_e({st},x),
\end{align*}
and since $1-\beta=\alpha^{m-1}$ and $\pi_d(\Lambda)=\beta$ for all $d\geq 1$,
\begin{align*}
P_{w,s} 
=& \pi_{D-d+1}(\Lambda)P_e(s,x)
	+ \sum_{T_0,T_1,\ldots,T_{m-1}\in {\mathcal T}(D-d)} 
	\pi_{D-d+1}(\cup_jT_j) 
	\prod_{t \in \cup_jT_j} P_e({st},x)\\
=& \pi_{D-d+1}(\Lambda)P_e(s,x)
	+ \sum_{T\in {\mathcal T}(D-d+1),T\neq\Lambda} 
	\pi_{D-d+1}(T) 
	\prod_{t \in T} P_e({st},x)\\
=& 
	\sum_{T\in {\mathcal T}(D-d+1)} 
	\pi_{D-d+1}(T) 
	\prod_{t \in T} P_e({st},x).
\end{align*}
This establishes~(\ref{eq:claim}) for all nodes
$s$ at depth $d-1$, completing the inductive
step and the proof of the theorem.
\qed

\subsection{Proof of Theorem~\ref{bctth}}

As the proof follows very much along the same lines as that 
of Theorem~3.2 of~\cite{BCT-JRSSB:22}, 
most of the details are omitted~here.
The proof is again by induction. First, we claim that:
\begin{equation}
P_{m,\lambda} = 
\max_{T \in {\mathcal T}(D)} 
p(x,T) = \max_{T \in {\mathcal T}(D)} 
\pi_{D}(T)\prod_{s \in T} P_{e} (s,x) .
\label{eq:pre-targetM}
\end{equation}
As in the proof of Theorem~\ref{ctwth}, in fact
we claim that the following more general 
statement holds: For any node $s$ at depth
$d$ with $0\leq d\leq D$, we have,
\begin{equation}
P_{m,s} = \max_{U \in {\mathcal T}(D-d)} \pi_{D-d}(U) 
\prod_{u \in U} P_{e} ({su},x),
\label{eq:claimM}
\end{equation}
where $su$ denotes the concatenation of contexts
$s$ and $u$. The proof of this is by an inductive step similar to that 
of Theorem~\ref{ctwth}. 
Taking $s=\lambda$ in~(\ref{eq:claimM}) implies~(\ref{eq:pre-targetM}). 

Then, it is sufficient to show that for the tree $T_1 ^*$ that is produced by the GBCT algorithm, $
P_{m,\lambda} = 
p(x,T^*_1).
$
This is again proved by induction, via an argument similar to the ones in 
the previous two cases. 
Finally, using~(\ref{eq:pre-targetM}) and dividing both sides with $p(x)$ 
gives that $\max_{T \in {\mathcal T}(D)} \pi (T|x) = \pi (T^*_1|x)$
and completes the proof of the theorem.
\qed

\subsection{The $k$-GBCT algorithm}

The $k$-BCT algorithm of~\cite{BCT-JRSSB:22} can be generalised in 
exactly the same manner as the CTW and BCT algorithms were generalised.
Its exact steps are not repeated here as the algorithm description
is quite lengthy.
The resulting algorithm identifies the top-$k$ {\em a posteriori} most 
likely context-tree models. The proof of the theorem claiming this is similar 
to the proof of Theorem~3.3 of~\cite{BCT-JRSSB:22} and thus also omitted. 
Again, the only important difference, both in the algorithm description and 
in the proof, is that the estimated probabilities $P_e(s,x)$ are used 
in place of their simple discrete version $P_e(a_s)$.

\subsection{Proof of Theorem~\ref{branchth}}

The proof follows along the same lines as that of 
Proposition 3.1 of~\cite{papag-K-pre:23}, again with the only difference 
that the new version of $P_e(s,x)$ needs to be used in place of their 
discrete version. Hence, most of the details are again omitted here. 

Note that every context-tree model $T  \in \mathcal {T} (D)$ can
be viewed as a collection of a number,~$\ell$, say, 
of $m$-branches, since every node in $T$ has either
zero or $m$ children. The proof is by induction on $\ell$. 
The result follows 
immediately from Theorem~\ref{ctwth} for $\ell = 0$,
since the only tree with no $m$-branches 
is $T = \{ \lambda \}$. For the inductive step, we assume the result 
holds for all trees that have $\ell$ $m$-branches, and suppose
that $T'\in \mathcal {T} (D)$ contains $(\ell + 1)$ $m$-branches 
and is obtained from some $T  \in \mathcal {T} (D)$ by adding a single
$m$-branch to one of its leaves.

\section{Proofs of Lemmas} \label{appB}

The proofs of these lemmas are mostly based on explicit computations. 
Recall that, for each context~$s$, the set~$B_s$ consists of those indices $i\in\{1,2,\ldots ,n\}$ such that the context of $x_i$   is~$s$. The important step in the following is the factorisation of the likelihood using the sets $B_s$.
In order to prove the lemmas for the AR model with parameters $\theta_s = (\boldsymbol \phi _ s , \sigma _ s ^2)$, we first consider an intermediate step 
in which the noise variance is assumed to be known and equal to~$\sigma ^2$. 

\subsection{Known noise variance}

Here, to any leaf $s$ of the context-tree model $T$, we associate an AR model with known variance $\sigma ^2$, so that,
\begin{equation}\label{ar_2}
x_n = \phi _ {s,1} x_{n-1} + \dots + \phi _ {s,p} x_{n-p} + e_n = {\boldsymbol \phi _ s} ^{\text{T}} \ \mathbf{ \widetilde{ x} } _{n-1} + e_n , \quad  e_n \sim \mathcal N (0, \sigma ^2 ).
\end{equation}
In this setting, the parameters of the model are only the AR coefficients $\theta_s \hspace*{-0.08 cm}= \hspace*{-0.03 cm} \boldsymbol \phi _ s$.\hspace*{-0.01 cm} For these, we use a Gaussian~prior, 
\begin{equation}
\theta _s\sim \mathcal N ( \mu _o , \Sigma _ o)  ,
\end{equation}
where $\mu _o , \Sigma _ o$ are  hyperparameters. Here,
the following expression can be derived
for the estimated probabilities~$P_e(s,x)$.

\begin{lemma}
\label{lem:known}
The estimated probabilities  $P_e(s,x)$ for the known-variance 
case are given by,
\begin{equation} \label{ar_pes_known}
P_e(s,x) = \frac{1}{(2 \pi \sigma ^2)^{|B_s| /2 }} \ \frac{1}{\sqrt{\text {det}( I + \Sigma _o S_3 / \sigma ^2})}  \  \exp{\bigg \{ - \frac {E_s} {2 \sigma ^2}\bigg \} },
\end{equation}
where $I$ is the identity matrix and $E_s$ is given by:
\begin{equation}\label{ar_pes_known2}
E_s = s_1 + \sigma ^2   \mu _o ^ {\text {T}} \Sigma _ o ^{-1} \mu _o  - ( \mathbf s_2  +  \sigma ^2  \Sigma _ o ^{-1} \mu _o )^ {\text {T}} (S_3 + \sigma ^ 2\Sigma _ o ^{-1} ) ^ {-1}  ( \mathbf s_2  +  \sigma ^2  \Sigma _ o ^{-1} \mu _o )   .
\end{equation}
\end{lemma}

\begin{proof}
For the AR model of~(\ref{ar_2}),
$$ p(x_i|T , \theta _s , x_{-D+1}^{i-1}) = \frac{1}{\sqrt{2 \pi \sigma ^2}} \  \exp \bigg \{ -\frac {1}{2 \sigma ^2} (x_i - {\theta _ s} ^{\text{T}} \mathbf{ \widetilde{ x} } _{i-1} ) ^2 \bigg \},$$
so that,
$$\prod _{i \in B_s}  p(x_i|T , \theta _s , x_{-D+1}^{i-1}) = \frac{1}{(\sqrt{2 \pi \sigma ^2})^{|B_s|}} \  \exp \bigg \{ -\frac {1}{2 \sigma ^2} \sum _{i \in B_s}(x_i -  {\theta _ s} ^{\text{T}} \mathbf{ \widetilde{ x} } _{i-1} ) ^2 \bigg \}.
$$
Expanding the sum in the exponent gives,
\begin{eqnarray*}
\sum _{i \in B_s}(x_i - {\theta _ s} ^{\text{T}} \mathbf{ \widetilde{ x} } _{i-1} ) ^2 
&=&
    \sum _{i \in B_s} x_i ^2 - 2  \theta_s ^{\text{T}} \sum _{i \in B_s} x_i \mathbf{ \widetilde{ x} } _{i-1} +  \theta_s ^{\text{T}} \sum _ {i \in B_s } \mathbf{ \widetilde{ x} } _{i-1}  \mathbf{ \widetilde{ x} } _{i-1} ^{\text{T}}  \theta_s\\
&=&
    s_1 - 2  \theta_s ^{\text{T}} \mathbf s_2 +  \theta_s ^{\text{T}} S_3  \theta_s, 
\end{eqnarray*}
from which we obtain that,
\begin{align*}
\prod _{i \in B_s}  p(x_i|T , \theta _s , x_{-D+1}^{i-1}) =& \frac{1}{(\sqrt{2 \pi \sigma ^2})^{|B_s|}} \  \exp \bigg \{ -\frac {1}{2 \sigma ^2} (s_1 - 2  \theta_s ^{\text{T}} \mathbf s_2 +  \theta_s ^{\text{T}} S_3  \theta_s  )  \bigg \}  \\
 =&  (\sqrt {2 \pi }) ^ p  \rho_s \ \mathcal N ( \theta_s ;\boldsymbol \mu , S )  ,
\end{align*}
by completing the square, where  $\boldsymbol \mu = S_3^{-1} \mathbf s_2$,  $ S = \sigma ^2 S_3 ^{-1}$, and,
\begin{equation}
\rho _s =\sqrt{ \frac{\text{det}(\sigma^2  S_3 ^{-1})}{(2 \pi \sigma ^2)^{|B_s|}}} \  \exp \bigg \{  -\frac {1}{2 \sigma ^2} (s_1 - \mathbf s_2 ^{\text{T}} S_3 ^{-1} \mathbf s_2) \bigg \} \ .
\end{equation}
So, multiplying with the prior:
\begin{equation*}
\prod _{i \in B_s}  p(x_i|T , \theta _s , x_{-D+1}^{i-1}) \pi ( \theta _ s)= (\sqrt {2 \pi }) ^ p  \rho_s \ \mathcal N ( \theta_s ;\boldsymbol \mu , S ) \ \mathcal N ( \theta _ s;  \mu _o , \Sigma _ o) =  \rho_s Z_s \ \mathcal N ( \theta_s ;\mathbf m , \Sigma ) ,
\end{equation*}
where $ \Sigma^{-1} = \Sigma _ o ^{-1} +  S ^{-1}, \ m = \Sigma \ (\Sigma _ o ^{-1} \mu_o + S^{-1} \boldsymbol \mu)$, and, 
\begin{equation}
Z _s =   \frac{1}{\sqrt{{\text {det} (\Sigma_o} +\sigma^2  S_3 ^{-1}) }} \  \exp{\bigg \{ -   \frac {1}{2} (\mu _o - S_3 ^{-1} \mathbf s_2)^{\text{T}} (\Sigma_o +\sigma^2  S_3 ^{-1}) ^{-1} ( \mu _o - S_3 ^{-1} \mathbf s_2)\bigg \} }  .
\end{equation}
Therefore,
\begin{equation}
\prod _{i \in B_s}  p(x_i|T , \theta _s , x_{-D+1}^{i-1}) \pi ( \theta _ s)= \rho_s Z_s \ \mathcal N ( \theta_s ;\mathbf m , \Sigma ), \label{13} 
\end{equation}
and hence,
$$
P_e(s,x) =  \int \prod _{i \in B_s}  p(x_i|T , \theta _s , x_{-D+1}^{i-1}) \ \pi (\theta _ s) \ d\theta_s \ = \rho_s Z_s .$$
Using standard matrix inversion properties, after some algebra the product $\rho_s Z_s $ can be rearranged to give exactly the required expression 
in~(\ref{ar_pes_known}).
\end{proof}

\subsection{Proof of Lemma~\ref{lem:PeAR}}

Now, we move back to the original case, as described in the main
text, where the noise variance is considered to be a parameter of the AR model,  so that $\theta_s = (\boldsymbol \phi _ s , \sigma _ s ^2)$. Here, the joint prior on the parameters is \mbox{$\pi(\theta _s) =\pi (\boldsymbol \phi _s | \sigma _s ^2) \pi (\sigma _ s ^2)$}, where,
\begin{align} \label{ar_pr_ap}
&\sigma _s ^2\sim \mbox{Inv-Gamma}(\tau , \lambda)    ,\\
&\boldsymbol \phi _s | \sigma _s ^2 \sim \mathcal N (\mu _o , \sigma_s ^2 \Sigma _o)  , \label{ar_pr2_ap}
\end{align}
and where $(\tau, \lambda,  \mu _o ,  \Sigma _o  )$ are the prior hyperparameters.
For the estimated probabilities $P_e(s,x)$, we just need to compute the integral:
\begin{align}
P_e(s,x) =&  \int \prod _{i \in B_s}  p(x_i|T , \theta _s , x_{-D+1}^{i-1}) \ \pi (\theta _ s ) \ d\theta_s  \\
=& \int \pi (\sigma _ s ^2) \left ( \int  \prod _{i \in B_s}  p(x_i|T , \boldsymbol \phi _ s , \sigma _ s ^2 , x_{-D+1}^{i-1}) \ \pi (\boldsymbol \phi _ s| \sigma _ s ^2) \ d\boldsymbol \phi _ s\right ) d\sigma_s ^2 .
\end{align}
The inner integral has the form of the estimated probabilities $P_e(s,x)$ from the previous section, where the noise variance was fixed. The only difference is that the prior $\pi (\boldsymbol \phi _ s | \sigma _ s ^2)$ of~(\ref{ar_pr2_ap}) now has covariance matrix~$\sigma _ s ^2 \Sigma _o $ instead of $\Sigma _o$. So, using~(\ref{ar_pes_known})-(\ref{ar_pes_known2}) with $\Sigma _o$ replaced by $\sigma _ s ^2 \Sigma _o $, yields, 
\begin{equation*}
P_e(s,x) =\int \pi (\sigma _ s ^2) \bigg \{ C_s ^ {-1}\bigg (\frac{1}{\sigma _ s ^2}  \bigg ) ^ {{|B_s|}/{2}} \exp \bigg ( - \frac{D_s}{2 \sigma _ s ^2} \bigg )  \bigg \} d\sigma_s ^2, 
\end{equation*}
with $C_s$ and $D_s$ as in Lemma~\ref{lem:PeAR}. 
And using the inverse-gamma prior $\pi (\sigma _ s ^2)$ of~(\ref{ar_pr_ap}), 
\begin{equation}\label{19}
P_e(s,x)= \ C_s ^ {-1} \  \frac {\lambda ^ {\tau}}{\Gamma (\tau)} \  \int  \bigg (\frac{1}{\sigma _ s ^2}  \bigg ) ^ {\tau ' +1 }  \exp \bigg ( - \frac{\lambda '}{ \sigma _ s ^2} \bigg )  d\sigma_s ^2 , 
\end{equation}
with $\tau ' = \tau + \frac{|B_s|}{2} $ and $\lambda ' = \lambda + \frac{D_s}{2}$.

\smallskip

The integral in~(\ref{19}) has the form of an inverse-gamma density with parameters $\tau ' $ and $\lambda '$, whose closed-form solution is, 
\[
P_e(s,x) =  C_s ^ {-1} \  \frac {\lambda ^ {\tau}}{\Gamma (\tau)} \ \frac {\Gamma (\tau ' )} {\left ( \lambda ' \right )^ {\tau '}} \  , 
\] 
which, as required, completes the proof the lemma. \qed

\subsection{Proof of Lemma~\ref{lem:piAR}}

In order to derive the required expressions for the posterior distributions of $\boldsymbol \phi _ s $ and $\sigma _ s ^2$, for a leaf~$s$ of model $T$, 
first consider the joint posterior distribution $\pi (\theta _ s | T, x) = \pi (\boldsymbol \phi _ s , \sigma _ s ^2 | T, x)$, given by,
\begin{align*}
\pi (\theta _ s | T, x) \propto  p(x| T, \theta _ s )   \pi (\theta_ s )   =    \prod _{i=1} ^ n  p(x_i|T , \theta _s ,  x_{-D+1}^{i-1})       \pi (\theta_ s ) 
 \propto \prod _{i \in B_s}   p(x_i|T , \theta _s ,  x_{-D+1}^{i-1})       \pi (\theta_ s ) , 
\end{align*} 
where we used the fact that, in the product, only the terms involving indices $i\in B_s$ are functions of~$\theta _s$. So,
\begin{align*}
\pi (\boldsymbol \phi _ s , \sigma _ s ^2 | T, x)  &\propto \left ( \  \prod _{i \in B_s}   p(x_i|T ,\boldsymbol \phi _ s , \sigma _ s ^2 , x_{-D+1}^{i-1})      \  \pi (\boldsymbol \phi _ s | \sigma _ s ^2)  \right ) \pi (\sigma _ s ^2 )  .
\end{align*}
Here, the first two terms can be computed from~(\ref{13}) of the previous section, where the noise variance was known. Again, the only difference is that we have to replace $\Sigma _o $ with $\sigma _ s ^2 \Sigma _o $ because of the prior $ \pi (\boldsymbol \phi _ s | \sigma _ s ^2)$ defined in~(\ref{ar_pr2_ap}). After some algebra, this gives, 
\begin{align*}
\pi (\boldsymbol \phi _ s , \sigma _ s ^2 | T, x)  &\propto  \bigg (\frac{1}{\sigma _ s ^2}  \bigg ) ^ {{|B_s|}/{2}} \exp \bigg ( - \frac{D_s}{2 \sigma _ s ^2} \bigg ) \  \mathcal N ( \boldsymbol \phi _ s ;\mathbf m_s , \Sigma _s )   \ \pi (\sigma _ s ^2 ) \ ,
\end{align*}
with $\mathbf m _s $  defined as in Lemma~\ref{lem:piAR}, 
and $\Sigma _ s = \sigma _ s  ^2 (S_3 + \Sigma _ o ^{-1}) ^{-1} $. 
 Substituting the prior $\pi (\sigma _ s ^2)$ in the last expression gives,
\begin {equation}\label{jointpost}
\pi (\boldsymbol \phi _ s , \sigma _ s ^2 | T, x)  \propto  \bigg (\frac{1}{\sigma _ s ^2}  \bigg ) ^ {\tau + 1 + {|B_s|}/{2}} \exp \bigg ( - \frac{\lambda + D_s/2}{\sigma _ s ^2} \bigg ) \  \mathcal N ( \boldsymbol \phi _ s ;\mathbf m_s , \Sigma _s )   . 
\end{equation}
From~(\ref{jointpost}), it is easy to integrate out $\boldsymbol \phi _ s $ 
and get the posterior density of $\sigma _s ^2$,
\[
\pi (\sigma _ s ^2 | T , x) = \int \pi (\boldsymbol \phi _ s , \sigma _ s ^2 | T, x) \ d \boldsymbol \phi _ s \propto  \bigg (\frac{1}{\sigma _ s ^2}  \bigg ) ^ {\tau + 1 + {|B_s|}/{2}} \exp \bigg ( - \frac{\lambda + D_s/2}{\sigma _ s ^2} \bigg ),
\]
which is of the form of an inverse-gamma distribution with parameters  $\tau ' = \tau + \frac{|B_s|}{2} $ and $\lambda ' = \lambda + \frac{D_s}{2}$, proving the first part of the lemma.

\smallskip

However, as $\Sigma _ s$ is a function of $\sigma _ s ^ 2$, integrating out $\sigma_s^2$ requires more algebra.
In specific, we have that,
\begin{align*}
\mathcal N ( \boldsymbol \phi _ s ;\mathbf m_s , \Sigma _s ) &\propto \frac {1}{\sqrt {\text {det} (\Sigma _s)}} \ \exp \bigg \{  -\frac {1}{2} (\boldsymbol \phi _ s -\mathbf m _ s) ^ {\text {T}} \Sigma _ s^ {-1} (\boldsymbol \phi _ s-\mathbf m _ s)  \bigg \}  \\
& \propto \bigg ( \frac {1} {\sigma _ s ^2} \bigg ) ^ {p/2}  \exp \bigg \{  -\frac {1}{2\sigma _ s ^2} (\boldsymbol \phi _ s-\mathbf m _ s) ^ {\text {T}} (S_3 + \Sigma _o ^{-1})  (\boldsymbol \phi _ s -\mathbf m _ s)  \bigg \} ,
 \end{align*}
and substituting this in~(\ref{jointpost}) gives
that $\pi (\boldsymbol \phi _ s  , \sigma _ s ^2 | T, x)$ is proportional to,
\begin{align*}
\bigg (  \frac {1} {\sigma _ s ^2}\bigg ) ^ {\tau +1 + \frac {|B_s|+ p} {2}} \hspace{-0.1 cm} \exp  \bigg \{  -\frac {1}{2\sigma _ s ^2}\bigg ( 2 \lambda + D_s + (\boldsymbol \phi _ s  -\mathbf m _ s) ^ {\text {T}} (S_3 + \Sigma _o ^{-1})  (\boldsymbol \phi _ s  -\mathbf m _ s) \bigg ) \bigg \} ,
 \end{align*}
which, as a function of $\sigma_s ^2$, has the form of an inverse-gamma density, allowing us to integrate out~$\sigma_s ^2$. 
Denoting \mbox{$L =   2 \lambda + D_s + (\boldsymbol \phi _ s-\mathbf m _ s) ^ {\text {T}} (S_3 + \Sigma _o ^{-1})  (\boldsymbol \phi _ s -\mathbf m _ s)$}, and $ \widetilde \tau = \tau +\frac {|B_s|+ p} {2}$, 
\begin{align*}
&\pi (\boldsymbol \phi _ s | T , x) = \int \pi (\boldsymbol \phi _ s  , \sigma _ s ^2 | T, x) \ d  \sigma _ s ^2 \propto \int  \bigg (  \frac {1} {\sigma _ s ^2}\bigg ) ^ {\widetilde \tau +1 }  \exp \bigg ( - \frac{L } {2 \sigma _ s ^2}\bigg ) \ d \sigma _ s ^2 = \frac {\Gamma (\widetilde \tau)} {(L/2 ) ^ {\widetilde \tau}} \ .
 \end{align*}
So, as a function of $\boldsymbol \phi _ s $, the posterior 
distribution $\pi(\boldsymbol \phi _ s |T,x)$ is,
\begin{align*}
\pi (\boldsymbol \phi _ s | T , x) \propto L^ {- \widetilde \tau} &= \bigg (  2 \lambda + D_s + (\boldsymbol \phi _ s -\mathbf m _ s) ^ {\text {T}} (S_3 + \Sigma _o ^{-1})  (\boldsymbol \phi _ s-\mathbf m _ s) \bigg ) ^ {- \frac {2 \tau + |B_s| + p}{2}} \\
& \propto \bigg ( 1 + \frac {1}{2 \tau + |B_s|} \  (\boldsymbol \phi _ s-\mathbf m _ s) ^ {\text {T}}\frac {(S_3 + \Sigma _o ^{-1})(2 \tau + |B_s|)}{(2 \lambda + D_s) }  (\boldsymbol \phi _ s -\mathbf m _ s)  \bigg ) ^ {- \frac {2 \tau + |B_s| + p}{2}} \\
& \propto \bigg ( 1 + \frac {1}{\nu} \  (\boldsymbol \phi _ s -\mathbf m _ s) ^ {\text {T}} P_s ^{-1}  (\boldsymbol \phi _ s -\mathbf m _ s)  \bigg ) ^ {- \frac {\nu+ p}{2}}  ,
\end{align*}
which is exactly in the form of a multivariate $t$-distribution, with $p$ being the dimension of $\boldsymbol \phi _ s$, and with $\nu, \mathbf m _ s$ and  $P _s $ exactly as given in Lemma~\ref{lem:piAR}, completing its proof. \qed

\subsection{Proof of Lemma~\ref{lem:ARCH}}
\label{appB3}

For the BCT-ARCH model, at every leaf $s$,
\begin{align}
    x_n \sim {\mathcal N}  (0, \sigma _ n ^ 2  ) , \; \; \; \; 
    \sigma _ n ^ 2  = \alpha _ {s,0} + \alpha _{s,1} x_{n-1} ^ 2 + \dots + \alpha _{s,p}x_{n-p} ^ 2 = \boldsymbol \alpha _ s ^{\text{T}} \ \mathbf{ { z } } _{n-1} ,
\end{align}
where $\theta _ s = \boldsymbol \alpha _ s = \left (  \alpha _ {s,0} , \alpha _{s,1}  , \dots , \alpha _{s,p} \right )^{\text{T}}$ and $\mathbf{ { z } } _{n-1} = (1, x_{n-1} ^ 2,\dots, x_{n-p} ^ 2 )^{\text{T}}.$ The proof of the lemma 
follows upon considering the log-likelihood of data with context $s$, given by,
\begin{equation} 
     L _ s  (  {\theta} _ s  ) =  \sum _{ i \in B_s} \log p (x_i | x_{-D+1} ^ {i-1} , \theta _ s )  = -  \frac {| B _ s |} {2}  \log (2 \pi)-  \frac {1}{2} \sum _ {i \in B _s} \left ( \log \sigma _ i ^ 2 + \frac {x _ i ^ 2}  {\sigma _i ^ 2} \right ) ,
\end{equation} 
and taking its derivatives with respect to $\theta_s$. As the dependence is implicit through $\sigma _ i ^ 2 = \theta _ s ^ {  \text {T}}  \ \mathbf{ { z } } _{i-1}$, taking the first derivative gives, 
\begin{equation}
    \frac{\partial L_s} {\partial \theta _s} =    \frac {1}{2} \sum _ {i \in B _s}  \frac {1}  {\sigma _i ^ 2} \left (\frac {x _ i ^ 2}  {\sigma _i ^ 2} -1 \right )  \frac{\partial \sigma _i ^ 2} {\partial \theta _s} =   \frac {1}{2} \sum _ {i \in B _s}  \frac {1}  {\sigma _i ^ 2} \left (\frac {x _ i ^ 2}  {\sigma _i ^ 2} -1 \right )  \mathbf{ { z } } _{i-1} ,
\end{equation}
and similarly taking the second derivative and its expectation finally gives,
\begin{equation}
\widehat {I} _ s =  \left \{  - \mathbb {E} \left (  \frac {\partial ^ 2 L _s}{ \partial \theta _ s ^ 2 }\right ) \right  \} =  \frac {1}{2} \sum _ {i \in B _s}  \left ( \frac {1}  {\sigma _i ^ 4 } \right )  \mathbf{ { z } } _{i-1}   \mathbf{ { z } } _{i-1} ^ {  \text {T}},
\end{equation}
completing the proof of the lemma. \qed

\section{Datasets} 
\label{list_of_data}

\subsection{\texttt{sim\_1}} \label{sim1}

This is a simulated dataset consisting of 
$n=600$ samples generated from a BCT-AR model with 
the context-tree model of Figure~\ref{tree}, 
a binary quantiser with threshold $c=0$, 
and AR order $p=2$. The complete specification of this BCT-AR model, 
also given in the main text, is,
\begin{align*} 
x_n \!=\! \left\{
\begin{array}{ll}
  0.7  \ x_{n-1} - 0.3  \ x_{n-2} + e_n,  \quad  e_n \sim \mathcal N (0, 0.15), \; \;  & \mbox{if} \ s = 1\text{:} \ \ \ x_{n-1}>0, \\
-0.3  \ x_{n-1} - 0.2  \ x_{n-2} + e_n,  \quad  e_n \sim \mathcal N (0, 0.10),  \; \;    &\mbox{if} \ s = 01\text{:} \   x_{n-1}\leq 0, \  x_{n-2}>0, \\
 0.5 \ x_{n-1} + e_n,  \quad \quad \quad \quad \quad \quad e_n \sim \mathcal N (0, 0.05),  \; \; & \mbox{if} \ s = 00\text{:} \ x_{n-1}\leq 0, \  x_{n-2}\leq 0.
\end{array}
\right.  
\end{align*}
Here, we also report the {\em evidence} $p(x|c,p)$ for a range of values of $c$ and $p$. Although maximising the evidence is a very common, well-justified Bayesian practice~\cite{rasmussen:00,mackay:92}, we report some values 
as a sanity check, to show that the evidence is indeed maximised at the 
true values of $c=0.0$ and $p=2$, confirming the effectiveness
of our inferential procedure for choosing $c$ and $p$.

\begin{table}[!h] 
  \caption{Using the evidence $p(x|c,p)$ to choose the AR order and the quantiser threshold.}
    \vspace{-0.15 cm}
  \centering
  \begin{tabular}{ccccccccccc}
\toprule
   & \multicolumn{5}{c}{AR order $p$}   &    \multicolumn{5}{c}{Threshold $c$}              \\
    \cmidrule(r){2-6}
\cmidrule(r){7-11}
 &   $1$ & $2$ & $3$ & $4$ & $5$ &  $-$0.1 & $-$0.05& 0& 0.05 & 0.1  \\
    \midrule

    $ - \log _ 2 p(x|c,p)  $    & 533 & \bf {519} & 526 &531 & 535 & 558 & 539 & \bf{519} & 555 & 577  \\
    \bottomrule
  \end{tabular}
\vspace{-0.3cm}
\end{table}

\newpage

\subsection{\texttt{sim\_2}} \label{sim2}

\begin{wrapfigure}{r}{0.38\linewidth}
  \begin{center}
  \vspace*{-1.4 cm}
    \includegraphics[width= 0.78 \linewidth, height= 0.78 \linewidth ]{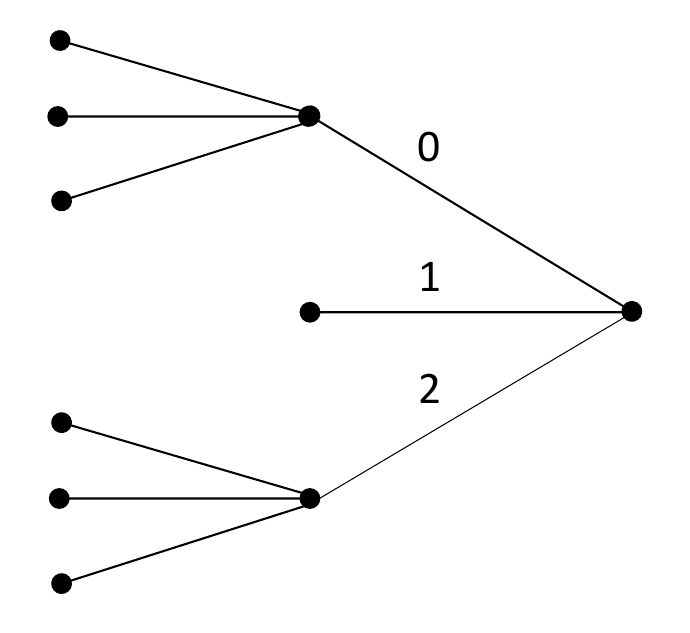}
  \end{center}
\vspace*{-0.65 cm}
\caption{Tree model of \texttt{sim\_2}.}
\label{tree_hmm}
\vspace*{-1.4 cm}
\end{wrapfigure}

This simulated dataset consists of $n=500$ samples that are generated from 
a BCT-AR model with respect to the ternary context-tree 
model in Figure~\ref{tree_hmm}. 
The thresholds of the quantiser are $c_1=-0.5$ and $c_2 = 0.5$, 
and the AR order is $p=1$. The complete specification of this BCT-AR model~is,
\begin{align*} 
x_n \!=\! \left\{
\begin{array}{ll}
0.5 \ e_n, \   & \mbox{if} \ s = 1,  \ 01, \  02, \  20 , \ 21,  \\
0.99 \  x_ {n-1} + 0.005 \ e_n, \   & \mbox{if} \ s = 00,  \ 22 ,\\
\end{array}
\right.  
\qquad
\end{align*}
with $e_n \sim \mathcal N (0, 1)$. 

\subsection{\texttt{sim\_3}} \label{sim3}

The third simulated dataset consists of $n=200$ samples generated 
from a SETAR model of order $p=5$, given by, 
\begin{align*} 
x_n \!=\! 
\left\{
\begin{array}{ll}
- 0.1 + 0.9  \  x _{n-1} +  0.9  \  x _{n-2} - 0.2 \ x _{n-5 }+  \ e_n, \   & \mbox{if}   \ x_ {n-1}  > -0.2, \\
0.2  + 0.1  \  x _{n-1} +  0.9  \ x _{n-5 }+  \ e_n, \   & \mbox{if}  \ x_ {n-1}  \leq -0.2, 
\end{array}
\right.  
\qquad
e_n \sim \mathcal N (0, 1).
\end{align*}

\subsection{\texttt{unemp}} \label{data_unemp}

This dataset consists of the $n=288$ values of the quarterly US
unemployment rate in the period from 1948 to 2019. 
It is publicly available from the US Bureau of Labor Statistics (BLS) 
at \url{https://data.bls.gov/timeseries/LNS14000000?years_option=all_years}.

\subsection{\texttt{gnp}} \label{data_gnp}

This is a time series of length $n=291$, corresponding to the quarterly US
Gross National Product (GNP) values between 1947 and 2019. It is available from the US Bureau of Economic Analysis (BEA), and can be retrieved from the Federal Reserve Bank of St. Louis (FRED) at \url{https://fred.stlouisfed.org/series/GNP}. Following~\cite{potter:95}, we consider the difference 
in the logarithm of the series, $y_n = \log x_n - \log x_{n-1}$.

For this dataset, the MAP BCT context-tree model 
is given in the main text, in Figure~\ref{tree_gnp}. It has depth~3, four states, $\mathcal {S} = \{0, 10, 110, 111 \}$, and posterior probability~42.6\%. The threshold of the binary quantiser selected using the procedure of Section~\ref{hyp} is $c=0.2$, so that $s=1$ if $y_{n-1}\geq 0.2$ 
and $s=0$ if $y_ {n-1} < 0.2 $. The selected AR order is $p=2$. The complete BCT-AR model with its MAP estimated parameters is given by,
\begin{align*} 
y_n \!=\! \left\{
\begin{array}{ll}
  1.16 + 0.71 \ y_{n-1} + 0.19 \ y_ {n-2} + 1.23 \ e_n, & \mbox{if}   \ s=0, \\
0.18 + 0.68 \ y_{n-1} - 0.26 \ y_ {n-2} + 1.19 \ e_n & \mbox{if}   \ s=10, \\
-1.05 + 1.40 \ y_{n-1} + 0.19 \ y_ {n-2} + 1.04 \ e_n & \mbox{if}   \ s=110, \\
 0.59 + 0.28 \ y_{n-1} + 0.31 \ y_ {n-2} + 0.75 \ e_n & \mbox{if}   \ s=111, \\
\end{array}
\right.  
\qquad
e_n \sim \mathcal N (0, 1).
\end{align*}

\subsection{\texttt{ibm}} \label{data_ibm}

This dataset consists of $n=369$ observations
corresponding to the 
daily IBM common stock closing 
price between May~17,~1961 and November 2, 1962. The data are taken 
from~\cite{box:book}, and are also 
available from the \texttt{R} package \texttt{fma}~\cite{fma}. 
The MAP context-tree model fitted to the dataset is shown in 
the main text in Figure~\ref{tree_ibm}. The complete BCT-AR model, 
with its MAP estimated parameters, is given by,
\begin{equation*}
x_n \!=\! \left\{
\begin{array}{ll}
 1.03  \ x_{n-1} - 0.03  \ x_{n-2} +  12.3  \ e_n,  \; \;  & \mbox{if} \  s=0, \\
1.17  \ x_{n-1} - 0.17  \ x_{n-2} + 6.86  \  e_n,   \; \;    &\mbox{if}  \ s=2 , \\
 -0.11  \ x_{n-1} + 1.11  \ x_{n-2} + 10.8  \  e_n,  \; \;  & \mbox{if} \  s=10, \\
1.22  \ x_{n-1} - 0.22  \ x_{n-2} + 5.32  \  e_n,   \; \;    &\mbox{if}  \ s=11 , \\
 0.15 \ x_{n-1} + 0.85  \ x_{n-2}+ 5.17  \  e_n,    \; \; & \mbox{if} \ s=12,
\end{array}
\right.  
\qquad e_n \sim \mathcal N (0, 1).
\end{equation*} 

\subsection{\texttt{ftse}} \label{appc7}

This is a dataset consisting of $n=7821$ 
daily observations
of the most commonly used UK-based stock market indicator, FTSE 100 (Financial Times
Stock Exchange 100 Index), for a time period of thirty years up to 7 April 2023.
It is available from 
Yahoo{\em !} Finance, at \url{https://finance.yahoo.com/quote/^FTSE/}. 

\subsection{\texttt{cac40}} \label{appc8}

This dataset consists of $n=7821$ daily observations of the 
most commonly used French stock~market index, CAC~40 
(Cotation Assist\'{e}e en Continu), for a period of thirty years up to 7 April 2023. It is available from Yahoo{\em !} Finance, at \url{https://finance.yahoo.com/quote/^FCHI/}. We consider the transformed time 
series, $y_n = 10  [\log x_ n -\log x_{n-1}]$.

The MAP context-tree model is given in the main text, in Figure~\ref{fig:bctarch}. It has depth~$3$, four leaves, $\mathcal {S}  = \{0, 10, 110, 111 \}$, 
and posterior probability~63.5\%. The complete BCT-ARCH model including the estimates of the parameters is given by,
\begin{align*}
    {\sigma} _n ^ 2 \!=\! \left\{
\begin{array}{ll}
 0.01 + 0.16 \ y_{n-1} ^ 2 +  0.17  \ y_{n-2} ^ 2 +  0.24  \ y_{n-3} ^ 2 +  0.14  \ y_{n-4} ^ 2 +  0.09  \ y_{n-5} ^ 2, \; \;  & \mbox{if} \  s=0, \\
 0.01 + 0.00 \ y_{n-1} ^ 2 +  0.21  \ y_{n-2} ^ 2 +  0.15  \ y_{n-3} ^ 2 +  0.13  \ y_{n-4} ^ 2 +  0.16  \ y_{n-5} ^ 2, \; \;  & \mbox{if} \  s=10, \\
 0.00 + 0.08 \ y_{n-1} ^ 2 +  0.04  \ y_{n-2} ^ 2 +  0.32  \ y_{n-3} ^ 2 +  0.11  \ y_{n-4} ^ 2 +  0.13  \ y_{n-5} ^ 2, \; \;  & \mbox{if} \  s=110, \\ 
  0.00 + 0.06 \ y_{n-1} ^ 2 +  0.09  \ y_{n-2} ^ 2 +  0.04  \ y_{n-3} ^ 2 +  0.15  \ y_{n-4} ^ 2 +  0.07  \ y_{n-5} ^ 2, \; \;  & \mbox{if} \  s=111.
\end{array}
\right. 
\end{align*}

\subsection{\texttt{dax}} \label{appc9}

This dataset  consists of $n=7821$ daily observations 
of the most commonly used German stock market index, 
DAX (Deutscher Aktienindex), for a period of thirty years up 
to 7 April 2023. It is available from 
Yahoo{\em !} Finance, at \url{https://finance.yahoo.com/quote/^GDAXI/}. 
We consider the transformed time 
series, $y_n = 10  [\log x_ n -\log x_{n-1}]$.
The MAP context-tree model is the same as for the CAC~40 index data, and
its posterior probability is now~48.6\%.
The estimated
ARCH coefficients are also very similar, with the complete BCT-ARCH model given by,
\begin{align*}
    {\sigma} _n ^ 2 \!=\! \left\{
\begin{array}{ll}
 0.01 + 0.14 \ y_{n-1} ^ 2 +  0.19  \ y_{n-2} ^ 2 +  0.22  \ y_{n-3} ^ 2 +  0.19  \ y_{n-4} ^ 2 +  0.12  \ y_{n-5} ^ 2, \; \;  & \mbox{if} \  s=0, \\
 0.01 + 0.00 \ y_{n-1} ^ 2 +  0.24  \ y_{n-2} ^ 2 +  0.19  \ y_{n-3} ^ 2 +  0.13  \ y_{n-4} ^ 2 +  0.16  \ y_{n-5} ^ 2, \; \;  & \mbox{if} \  s=10, \\
 0.01 + 0.02 \ y_{n-1} ^ 2 +  0.05  \ y_{n-2} ^ 2 +  0.28  \ y_{n-3} ^ 2 +  0.06  \ y_{n-4} ^ 2 +  0.10  \ y_{n-5} ^ 2, \; \;  & \mbox{if} \  s=110, \\ 
  0.01 + 0.00 \ y_{n-1} ^ 2 +  0.08  \ y_{n-2} ^ 2 +  0.04  \ y_{n-3} ^ 2 +  0.13  \ y_{n-4} ^ 2 +  0.14  \ y_{n-5} ^ 2, \; \;  & \mbox{if} \  s=111.
\end{array}
\right. 
\end{align*}

\subsection{\texttt{s\&p}} \label{appc10}

Finally, this dataset consists of $n=7821$ daily observations of 
Standard and Poor's~500 Index (S\&P~500), for a period of thirty 
years up to 7 April 2023. It is available from Yahoo{\em !} Finance, 
at \url{https://finance.yahoo.com/quote/^GSPC/}. Again we consider 
the transformed time 
series, $y_n = 10  [\log x_ n -\log x_{n-1}]$.
The MAP context-tree model is given in the main text, in Figure~\ref{fig:bctarch}. It has depth~$3$, five leaves, $\mathcal {S}  = \{00, 01, 10, 110, 111 \}$, 
and posterior probability~91.4\%. The complete BCT-ARCH model
is given by,
\begin{align*}
    {\sigma} _n ^ 2 \!=\! \left\{
\begin{array}{ll}
 0.00 + 0.12 \ y_{n-1} ^ 2 +  0.43  \ y_{n-2} ^ 2 +  0.17  \ y_{n-3} ^ 2 +  0.29  \ y_{n-4} ^ 2 +  0.13  \ y_{n-5} ^ 2, \; \;  & \mbox{if} \  s=00, \\
 0.00 + 0.20 \ y_{n-1} ^ 2 +  0.05  \ y_{n-2} ^ 2 +  0.15  \ y_{n-3} ^ 2 +  0.18  \ y_{n-4} ^ 2 +  0.19  \ y_{n-5} ^ 2, \; \;  & \mbox{if} \  s=01, \\
 0.00 + 0.06 \ y_{n-1} ^ 2 +  0.27  \ y_{n-2} ^ 2 +  0.19  \ y_{n-3} ^ 2 +  0.16  \ y_{n-4} ^ 2 +  0.11  \ y_{n-5} ^ 2, \; \;  & \mbox{if} \  s=10, \\
 0.00 + 0.09 \ y_{n-1} ^ 2 +  0.10  \ y_{n-2} ^ 2 +  0.24  \ y_{n-3} ^ 2 +  0.14  \ y_{n-4} ^ 2 +  0.15  \ y_{n-5} ^ 2, \; \;  & \mbox{if} \  s=110, \\ 
  0.00 + 0.01 \ y_{n-1} ^ 2 +  0.14  \ y_{n-2} ^ 2 +  0.03  \ y_{n-3} ^ 2 +  0.20  \ y_{n-4} ^ 2 +  0.12  \ y_{n-5} ^ 2, \; \;  & \mbox{if} \  s=111.
\end{array}
\right. 
\end{align*}

\subsection{Datasets from Section~\ref{sec:stat_sig}} \label{appc11}

The datasets used in Section~\ref{sec:stat_sig} to judge the statistical significance of the volatility forecasting results consist of a total of $N=21$ major stock market indices: the S\&P~500 which is a major US stock market index, together with 20 major European stock market indices (including FTSE~100, CAC 40 and DAX as before). Each dataset consists of a thirty year period of daily observations, corresponding to a maximum of $n=7821$ observations for each index (or at least, as many of those as were available online, either from Yahoo{\em !} Finance, 
at \url{https://finance.yahoo.com}, or from the Wall Street Journal, at \url{https://www.wsj.com/market-data/stocks/emea}). Similarly with Section~\ref{arch_forecasting}, a total of 130 observations -- corresponding to half a year of trading days -- is used as the test set in each case, with 7 April 2023 being chosen as the final day of the test set for half the datasets (as above), and 6 March 2025 being chosen as the final day for the second half of them, in order to  include more recent dates as well. 

All the training procedures carried out in this section are identical with Section~\ref{arch_forecasting}, as detailed also in Section~\ref{train_details} below. In most cases, the MAP context-tree model is either one of the trees shown in Figure~\ref{fig:bctarch} for the main stock market indices (FTSE 100, CAC 40, DAX and S\&P~500), or a small modification of one of these trees, which again gives a rich picture and a nice interpretation for the underlying volatility asymmetries present in the data. The log predictive density is used for evaluating forecasting performance exactly as in Section~\ref{arch_forecasting}.

\section{BCT-ARCH results on simulated data}
\label{app_sim_arch}

Here we present some more detailed examples 
of the performance of the BCT-ARCH methods on simulated data.
In particular, we examine the accuracy of the
approximations of~(\ref{laplace}) for the estimated
probabilities $P_e(s,x)$. Their accuracy requires two
things: First, the values of the 
iterates $\widehat{\theta}^{ \ (k)}_s$ need to be 
sufficiently close to their limiting values, namely,
the maximum likelihood estimates $\widehat{\theta}_s$.
This is easy to check by trying different initialisations 
and allowing for a large enough number of iterations,~$M$. 
Second, there need to be enough datapoints associated to 
each node $s$ of $T _ {\text{MAX}}$ for
the Laplace approximations of~(\ref{laplace}) to be close 
enough to the integrals~in~(\ref{pes}). This can be ensured 
by checking that the maximum context depth~$D$ is small enough 
compared to the total number of observations,~$n$.
 [It is noted that this
condition would not be required with the MCMC approach of~\cite{chib:95,chib:01}, but at the expense of higher
computational~complexity].

In the examples below we find that,
for datasets of length comparable to those studied in Section~\ref{s:volatility},
the choices $M=10$ and $D=5$ 
satisfy both the above~conditions for the approximations to be good enough.

First we consider data generated by the model considered in the
intentionally `difficult' example of Section~\ref{s:simulated}.
The posterior probability of the true context-tree
model, $\pi(T^*|x)$, is shown 
in Table~\ref{table_sim_arch}
as a function of different choices for $M$
and $D$.

\begin{table}[!h]
  \centering
  \caption{Posterior probability of the true underlying model, $\pi (T^*|x)$, 
	as more data become available.}
  \vspace{-0.2 cm}
\label{table_sim_arch}
  \begin{tabular}{lcccc}
\midrule
 & $n =1000$ & $n=2500$ & $n=5000$ & $=10000$ \\
 \midrule
 $D=3, M =10$ & 0.07 & 0.43  & 0.89 & 0.99 \\
 $D=4, M =10$ & 0.07 & 0.43  & 0.90 & 0.99 \\
 $D=5, M =10$ & 0.00 & 0.42  & 0.90 & 0.99 \\
  $D=5, M =100$ & 0.00 & 0.42  & 0.90 & 0.99 \\
    \midrule
  \end{tabular}
\end{table}

For all values of $M$ and $D$, with $n=1000$ observations the 
MAP context-tree model is the empty tree, corresponding to a single ARCH model. 
The posterior 
probability of the true tree model is still 
small as there are not enough 
data to support a more complex model. For $D=5$, there is a 
small difference in $\pi(T^*|x)$ compared to $D=3,4$, which might 
suggest that there are also not enough data yet for the Laplace 
approximations of~(\ref{laplace}) to be accurate. With $n=2500$ 
observations the MAP context-tree model is now $T^*$, with posterior 
probability $\pi(T^*|x) \approx 0.42$, and with all values of~$D$ 
effectively giving identical results. With $n=5000$ and $n=10000$ observations,
the posterior probability of the true model becomes, respectively,
$0.90$ and~$0.99$, for all values of $D$. Increasing the number of Fisher 
iterations from $M=10$ to $M=100$, and trying different initialisations 
gave identical results. 
We also recall from 
Section~\ref{s:simulated} that the estimates of the parameters
based on $n=5000$ observations with $M=10$ and $D=5$ were very close
to their true underlying values; recall~(\ref{eq:MAPs1}) and~(\ref{eq:MAPs2}).

As a second example, we consider the binary context-tree model with 
depth~$2$ and states $\mathcal {S} = \{1, 01, 00 \}$. The true model (left) and the model fitted from $n=5000$ observations with $M=10$ and $D=5$~(right) are given below, with the posterior probability 
of the true context-tree model being $\pi(T^*|x) \approx 0.98$.
\begin{align*} \label{bctarrch_sim2}
\sigma _n ^ 2 \!=\! \left\{
\begin{array}{ll}
 0.20 + 0.30  \ x_{n-1} ^ 2 +  0.10  \ x_{n-2} ^ 2   \\
 0.10 + 0.20  \ x_{n-1} ^ 2 +  0.10  \ x_{n-2} ^ 2   \\
 0.10 + 0.20  \ x_{n-1} ^ 2   
\end{array}
\right. ,  \; \;
\widehat {\sigma} _n ^ 2 \!=\! \left\{
\begin{array}{ll}
 0.20 + 0.30  \ x_{n-1} ^ 2 +  0.13  \ x_{n-2} ^ 2, \; \;  & \mbox{if} \  s=00, \\
  0.10 + 0.21  \ x_{n-1} ^ 2 +  0.14  \ x_{n-2} ^ 2, \; \;  & \mbox{if} \  s=01, \\
 0.10 + 0.20  \ x_{n-1} ^ 2  +  0.00  \ x_{n-2} ^ 2 , \; \;  & \mbox{if} \  s=1.
\end{array}
\right. 
\end{align*}

Overall, we conclude that the choices
$M=10$ and $D=5$ lead to very effective
inference with the BCT-ARCH model
for time series consisting
of around $5000-10000$ observations,
such as those studied in Section~\ref{s:volatility}.

\section{Training details} \label{train_details}

\subsection{BCT-AR experiments}

Here we specify the training details for all the methods used in the 
forecasting experiments of Section~\ref{experiments}. In all the examples 
the training set consists of the first 50\% of the observations.
All methods are updated at every timestep in the test data.

\medskip

\noindent
{\bf BCT-AR.} The hyperparameters are chosen as
described in Section~\ref{hyp}.
The default value of $\beta= 1-2^{-m+1}$ is taken
for the BCT prior~\cite{BCT-JRSSB:22}, 
the maximum context depth is $D=10$, and the AR order and 
the quantiser thresholds are selected at the end 
of the training set by maximising the evidence.
The MAP context-tree model and the MAP parameter
estimates are used for forecasting, and they 
are updated at every timestep in the test data.

\medskip

For each of the remaining methods, a range of parameter
values is considered, the entire forecasting experiment
is carried out for each candidate set of parameters,
and only the best result is reported in the main text.

\medskip

\noindent
{\bf ARIMA} and {\bf ETS}. The \texttt{R} package \texttt{forecast}~\cite{hyndman:08}
is used, together with the automated functions \texttt{auto.arima} and \texttt{ets} for fitting ARIMA and ETS models, respectively.

\medskip

\noindent
{\bf SETAR.} 
The \texttt{R} package \texttt{TSA}~\cite{tsa} is used
in conjunction with the commonly used conditional least 
squares method of~\cite{chan:93}. 
AR orders between $p=1$ and $p=5$
and delay parameter values between $d=1$ and $d=5$
are considered.

\medskip

\noindent
{\bf MAR.}
The \texttt{R} package \texttt{mix-AR}~\cite{mixar} is used.
Mixture models with $K=2$ and $K=3$ components, with AR orders between $p=1$ and $p=5$ are considered.

\medskip

\noindent
{\bf MSA.}
The \texttt{R} package \texttt{MSwM}~\cite{mswm} is 
used, which uses the EM algorithm for estimation. 
Models with $K=2$ and $K=3$ hidden states and AR orders between $p=1$ and $p=5$ are considered.

\medskip

\noindent
{\bf NNAR.}
The \texttt{R} package \texttt{forecast}~\cite{hyndman:08}
is used with the function \texttt{nnetar}.
AR orders between $p=1$ and $p=5$
are considered.

\medskip

\noindent
{\bf DeepAR} and {\bf N-BEATS.}
The implementations in 
the Python library GluonTS~\cite{alexandrov:20}
are used.
As the computational cost per 
iteration is different for the two methods, 
we use slightly different numbers of epochs and batches-per-epoch for 
each of them, in order to give similar empirical running times.
For DeepAR we use 5 epochs with 50 batches/epoch, and for N-BEATS 
we use 3 epochs with 20 batches/epoch.
AR orders between $p=1$ and $p=5$ are considered for both.

\subsection{BCT-ARCH experiments}

Here we specify the training details for all the methods used in the 
forecasting
experiments of Section~\ref{arch_forecasting}. 
For all four datasets,
the test set consists of the last 130 observations. All
methods are updated at every timestep in the test data.

\medskip

\noindent
{\bf BCT-ARCH. } A binary quantiser with threshold 
$c=0$ in used, and the default value $\beta=0.5$ 
corresponding to $m=2$ is chosen for the BCT prior. 
The maximum 
context depth $D$ and the ARCH order $p$
are both taken equal to~5, and the number of Fisher iterations
is taken to be $M=10$; cf.~Section~\ref{app_sim_arch} above.
The MAP context-tree model identified by the GBCT algorithm and the 
estimates for the ARCH parameters 
obtained from the scoring algorithm
are used for forecasting, and they 
are updated at every timestep in the test data.

\medskip

\noindent
{\bf GARCH}, {\bf GJR}, and {\bf EGARCH.} 
The \texttt{R}~package \texttt{rugarch}~\cite{rugarch} is used, 
together with the automated specifications \texttt{sGARCH}, \texttt{gjrGARCH}, and \texttt{eGARCH}.

\medskip

\noindent
{\bf MSGARCH.}
The \texttt{R}~package \texttt{MSGARCH}~\cite{ardia:19} is used.
Forecasting is performed using $K=2$ and $K=3$ hidden states,
with the best result reported for each dataset.

\medskip

\noindent
{\bf SV.}
The \texttt{R}~package 
\texttt{stochvol}~\cite{HK:21} is used, which implements MCMC 
samplers for Bayesian inference. At every timestep in the 
test set we use 1000 MCMC samples from the desired predictive density, after an initial burn-in of 100 samples. 

\section{Empirical running times} \label{run_times}

As discussed in Sections~\ref{compl}
and~\ref{s:compl2}, an important advantage of the BCT-X
framework is that its associated algorithms 
have low computational complexity and
allow for 
very efficient sequential updates, making it very practical 
for online forecasting applications.  

\subsection{BCT-AR experiments}

In sharp contrast with the computational efficiency of
BCT-AR forecasting, DeepAR and N-BEATS do not allow 
for incremental training, 
so the models are re-trained in each timestep from scratch.
This is computationally very costly as it involves 
gradient optimisation. 
The complexity of the classical statistical approaches 
lie somewhere between that of BCT-AR and the RNN
approaches:  They need to 
be re-trained at every timestep,
but the cost required per timestep is 
lower than that of the RNN methods; see also~\cite{makridakis:18b}.
Finally, ss discussed in Section~\ref{mar_comp}, the MSA model has much 
higher computational requirements compared to the other classical 
statistical methods, as the presence of the hidden state process 
makes inference much harder. 
From Table~\ref{table_emp} it is observed that
the BCT-AR model clearly outperforms 
all the benchmarks in terms of empirical running times. 
All experiments were carried out on a common laptop.  

\begin{table}[!h]
  \centering
  \caption{Empirical running times in the BCT-AR experiments*.}
  \vspace{-0.25 cm}
\label{table_emp}
  \begin{tabular}{lccccccccccc}
\midrule
 & \hspace* { -0.4 cm}  BCT-AR \hspace* { -0.2 cm} & ARIMA  \hspace* { -0.2 cm} & ETS \hspace* { -0.2 cm}  & NNAR \hspace* { -0.2 cm}  & DeepAR \hspace* { -0.2 cm}  & N-BEATS  \hspace* { -0.2 cm}  & MSA \hspace* { -0.2 cm}  & SETAR \hspace* { -0.2 cm}  & MAR \hspace* { -0.23 cm}  \\
 \midrule
\texttt {sim\_1} &  \hspace* { -0.4 cm}  \bf 7.4~s & 68~s &  31~s  & 4.9~min  & 2.4~h & 7.4~h & 45~min & 81~s & 19~min \hspace* { -0.23 cm}  \\
\texttt {sim\_2} &  \hspace* { -0.4 cm}  \bf 8.1~s & 61~s &  23~s  & 3.2~min  & 2.1~h & 6.3~h & 24~min & 98~s & 2.5~min \hspace* { -0.23 cm} \\
\texttt {sim\_3} &  \hspace* { -0.4 cm}  \bf 2.4~s & 28~s &  5.2~s  & 48~s  & 1.0~h & 4.0~h & 12~min & 11~s & 2.7~min \hspace* { -0.23 cm}  \\
    \midrule
  \texttt {unemp} &  \hspace* { -0.4 cm}  \bf 3.1~s & 42~s &  11~s  & 69~s  & 1.0~h & 4.1~h & 16~min & 17~s & 6.4~min \hspace* { -0.23 cm}  \\ 
    \texttt {gnp} &  \hspace* { -0.4 cm}  \bf 2.2~s & 80~s &  10~s  & 91~s  & 1.5~h & 5.2~h & 17~min & 19~s & 6.6~min \hspace* { -0.23 cm}  \\
      \texttt {ibm} &  \hspace* { -0.4 cm}  \bf 4.6~s & 58~s &  16~s  & 32~s  & 2.2~h & 5.3~h & 22~min & 28~s & 7.6~min \hspace* { -0.23 cm}  \\ 
    \midrule
  \end{tabular}
\end{table}

\vspace{-0.25 cm}

\subsection{BCT-ARCH experiments}

Similarly, the BCT-ARCH model is found to have the best
performance among all the alternatives in terms of empirical 
running times. The family of 
ARCH models and their extensions follow, as they need to 
be re-trained at every timestep.
Finally, the MSGARCH and SV models 
have by far the largest computational requirements,
as in these cases the randomness present 
in the hidden state process makes inference much more 
expensive; 
see also Section~\ref{garch_comp}.

\begin{table}[!ht]
  \centering
  \caption{Empirical running times in the BCT-ARCH experiments*.}
  \vspace{-0.25 cm}
\label{table_arcH_run}
  \begin{tabular}{lccccccc}
\midrule
 & BCT-ARCH  &  ARCH   &   GARCH  &  GJR  &  EGARCH   &   \  MSGARCH  &  SV  \\
 \midrule
 \texttt {ftse} & \bf 4.1 s & 5.5 min & 2.5 min & 4.8 min & 5.0 min & 35 min & 1.1 h \\
  \texttt {cac40} & \bf 4.8 s & 6.3 min & 3.3 min & 7.2 min & 5.9 min & 47 min & 1.2 h \\
  \texttt {dax} & \bf 4.5 s & 6.9 min & 3.0 min & 5.9 min & 6.2 min & 39 min & 1.0 h \\
  \texttt {s\&p} & \bf 4.2 s & 5.2 min & 2.6 min & 4.6 min & 4.7 min & 43 min & 1.2 h \\
    \midrule
  \end{tabular}
\end{table}

*The empirical running times of all benchmarks could be reduced if 
the models were not re-trained at every timestep,
but that would incur a degradation of prediction accuracy.

\end{document}